\documentclass[11pt,reqno]{amsart}

\usepackage[english]{babel}
\usepackage[utf8]{inputenc}
\usepackage{bbm}
\usepackage{amscd}
\usepackage{amsmath,amssymb}
\usepackage{hyperref}

\numberwithin{equation}{section}
\usepackage[margin=2.9cm]{geometry}
\usepackage{epstopdf} 

\newcommand{\bP}{\mathbb{P}}
\newcommand{\bQ}{\mathbb{Q}}
\newcommand{\bL}{\mathbb{L}}

\newcommand{\bE}{\mathbb{E}}
\newcommand{\bR}{\mathbb{R}}
\newcommand{\bN}{\mathbb{N}}

\newcommand{\bT}{\mathbb{T}}

\usepackage[pdftex]{color}

\newtheorem{theorem}{Theorem}[section]

\newtheorem{prop}[theorem]{Proposition}
\newtheorem{assumption}{Assumption}

\theoremstyle{definition}
\newtheorem{definition}[theorem]{Definition}

\theoremstyle{remark}
\newtheorem{remark}[theorem]{Remark}
\newtheorem{example}[theorem]{Example}

\def\eqref#1{(\ref{#1})}

\begin{document}

\title{Additive energy forward curves in a Heath-Jarrow-Morton framework}
\thanks{
This work was partially done while M. Piccirilli visited the University of Oslo. F. E. Benth acknowledges financial support from the research project FINEWSTOCH funded by the 
Norwegian Research Council. T. Vargiolu acknowledges financial support from the research project CPDA158845 ``Multidimensional polynomial processes and applications to new challenges in mathematical finance and in energy markets", funded by the University of Padova.}
\date{\today}

\author[F. E. Benth]{Fred Espen Benth} 
\address{University of Oslo, Department of Mathematics, PO Box 1053 Blindern, N-0316 Oslo (Norway)} 
\email{fredb@math.uio.no}
\author[M. Piccirilli]{Marco Piccirilli}
\address{University of Padova, Department of Mathematics, via Trieste 63, Torre Archimede, I-35121 Padova (Italy).}
\email[Corresponding author]{mpicciri@math.unipd.it}
\author[T. Vargiolu]{Tiziano Vargiolu}
\address{University of Padova, Department of Mathematics, via Trieste 63, Torre Archimede, I-35121 Padova (Italy)}
\email{vargiolu@math.unipd.it}



\begin{abstract}
One of the peculiarities of power and gas markets is the delivery mechanism of forward contracts. The seller of a futures contract commits to deliver, say, power, over a certain period, while the classical forward is a financial agreement settled on a maturity date. Our purpose is to design a Heath-Jarrow-Morton framework for an additive, mean-reverting, multicommodity market consisting of forward contracts of any delivery period. Even for relatively simple dynamics, we face the problem of finding a density between a risk-neutral measure $\bQ$, such that the prices of traded assets like forward contracts are true $\bQ$-martingales, and the real world probability $\bP$, under which forward prices are mean-reverting. By assuming that forward prices can be represented as affine functions of a universal source of randomness, we can completely characterize the models which prevent arbitrage opportunities. 
In this respect, we prove two results on the martingale property of stochastic exponentials. The first allows to validate measure changes made of two components: an Esscher-type density and a Girsanov transform with stochastic and unbounded kernel. 
The second uses a different approach and works for the case of continuous density. We show how this framework provides an explicit way to describe a variety of models by introducing, in particular, a generalized Lucia-Schwartz model and a cross-commodity cointegrated market.
\end{abstract}

\maketitle

\section{Introduction}

Since their deregulations, which took place in many countries over the last few decades, energy markets are rapidly evolving sectors and are bringing to the attention of practitioners and researchers challenging problems from the modeling perspective. The most active segment is often the forward market and, as such, the most liquid derivative products are forward contracts. 
We do not make distinctions between forwards and futures, since the results are the same in the case of deterministic interest rate, as we assume here. 
Throughout this paper, we adopt the same terminology as \cite{Benth20081116}, and reserve the name \emph{forward} to contracts with delivery at a fixed future time, while agreements to deliver the commodity over a period are called \emph{swaps}. The latter category is important especially in modeling electricity and natural gas markets, where the commodity is exchanged, either physically or financially, over a certain time period (e.g. a day, month or a whole year). A thorough study of mathematical models for these markets, together with a description of their most salient empirical features, can be found, for instance, in \cite{MR2416686,serletis}. 

This paper aims at developing a consistent and tractable framework for a multicommodity energy forward market by applying the Heath-Jarrow-Morton paradigm \cite{hjm}. Intuitively, this consists in describing forward prices as stochastically evolving functions of time and delivery dates, across different markets. To the best of our knowledge, the first work to apply this idea to energy-related commodities is \cite{clewlow} and since then many other works have followed, among others \cite{koeke,MR2241631,weron2006,kieselSB}; for a review of HJM models in power and gas markets and further references see \cite[Chapter 6]{MR2416686}. 
In this paper, we want to design a realistic forward market model with the same philosophy of \cite{Benth20081116}, where it is theoretically possible to trade contracts with any delivery period and no-arbitrage relations must hold among them. 
We propose a mean-reverting stochastic process (indeed, an Ornstein-Uhlenbeck-type of model parameterized over delivery times) for the forward and swap price dynamics. By formulating the model under the market probability measure $\bP$, we can represent the stylized empirical behavior of observed prices.
The no-arbitrage constraints among forwards and swaps are established via explicit relations between the parameters in the respective dynamics. 
As pointed out in \cite{Benth20081116}, specifying a stochastic evolution for the forward curve and then deriving the swap price as the average over the delivery times has the disadvantage of losing desirable distributional features and, in some cases, even the Markov property. This results in non-tractable models, which inherit a complex probabilistic structure. Therefore, in that paper the authors argue in favor of a swap market model, where only the so-called atomic swaps are directly modeled. Though prices of non-atomic swaps can be reconstructed by arbitrage arguments, it is not possible to use all the available information in the market when fitting the model  on real data. These facts motivate us to look for a HJM market model general enough to include both forwards and swaps, still preserving tractability of the resulting stochastic structure.

In this paper we study dynamical models which are additive, meaning that we do not perform a logarithmic transformation of prices (see Section 2 for more details). Recently, additive models have gained popularity, especially for describing power spot prices, e.g. \cite{MR2323278, fanone, gallana,HWK, kieselpara, latini}. This is due to their ability of reproducing rather well the stylized features of electricity prices, including the empirical evidence of negative spot prices, and providing explicit formulas for derivative pricing. Generally the problem, when these models are used for commodities other than power, is that negative prices can occur. Nevertheless, according to one's modeling preferences, it is possible e.g. to resort to subordinators (see \cite{MR2323278}). 
The multi-commodity setting has not been extensively studied for energy forward markets applications (see, however, the spot-based models by \cite{farkas,jaimungal,paschke}).
With the gradual integration of power markets, say, there is a need for cross-commodity dynamical models which can describe the price evolutions in different power markets simultaneously, such as the Nordic NordPool power and the German EEX market. 
In fact, our multicommodity framework allows us to develop models where commodities have various kinds of dependencies, like correlations among the driving processes or cointegration/price couplings among them (see Section 5). In addition, it opens the way for computing in a more realistic way the prices of multicommodity derivative assets, as dark or spark spreads written either on spot or futures prices and, on the risk management side, implementing consistently risk measures like PaR and VaR on multicommodity portfolios.

Our main idea is to express forward prices as affine transformations of a universal source of randomness, the latter being independent of the delivery date (see Assumption \ref{18}). This simplifying assumption allows both to preserve Markovianity and to describe consistently the related swap price processes. Most importantly, this will be the key to proving the existence of equivalent martingale measures $\bQ$, ensuring arbitrage-free models. In fact, the presence of mean-reversion compels us to face non-trivial mathematical hurdles. In this regard, we prove in Theorem \ref{85} the martingale property of stochastic exponentials 
where the Lévy part is of Esscher-type, while the Girsanov kernel of the Brownian component is affine in the state variable and thus, in particular, stochastic. Then, we give a different proof in the case of continuous density and continuous kernel, by applying a weak Novikov-type condition on the series representation of the exponential function (cf. Theorem \ref{7}). The proof of this result relies on the asymptotic properties of Gaussian moments. As we move on to consider more general L\'{e}vy kernels, we see that the same technique can not be applied. This is shown to be related to the fact that various examples of infinitely divisible distributions, except the Gaussian distribution, do not satisfy the needed moments asymptotics.

We demonstrate the utility of our theoretical framework by specifying two exemplary models. The first, which is presented in Section 4, is a generalization of the additive two-factor Lucia-Schwartz model \cite{Lucia2002}. We extend it by introducing a mean-reverting arbitrage-free forward dynamics, which is capable to describe a finer volatility term structure. This allows us to account for seasonal effects in price variations, as is typically observed, for instance, in power or gas markets. A calibration procedure and an empirical application to the German power futures market of a version of this model has been carried out in a parallel work by \cite{latini}. Furtherly, we introduce a multidimensional model for a mean-reverting cointegrated forward market that respects the no-arbitrage constraints. In particular, we see how these constraints imply certain conditions on the mean-reversion coefficients of the futures curve dynamics (see Section 5). 

The paper is structured as follows. Section 2 describes our application of the HJM approach to a multicommodity, additive, energy forward market. In Section 3 we study how our main assumption identifies the dynamics of the processes and the set of equivalent martingale measures. Then, we prove the two main results on the martingale property of stochastic exponentials. A generalization of the two-factor Lucia-Schwartz model \cite{Lucia2002} is presented in Section 4. Finally, in Section 5 we develop a cross-commodity model with cointegrated dynamics.
In Section 6 we make some final remarks. Appendix A contains the proof of Theorem \ref{85}, while the proof of Theorem \ref{7} is presented in Appendix B.

\section{An HJM-type Market for Energy Forward Contracts}

In this section we introduce the general structure of our forward market by applying the Heath-Jarrow-Morton approach to energy forward curves. As already mentioned, we assume that it is possible to trade contracts for one or more commodities in the form of forwards with instantaneous maturity and swaps with arbitrary delivery period. 
In this way we are able to introduce a sufficiently general framework to be applied to energy-related markets in a consistent way. Both the forward and swap dynamics are described by means of stochastic differential equations and the relations among them follow from natural no-arbitrage conditions. 
We extend the analysis by \cite[Sections 3 and 4]{Benth20081116} to a multidimensional framework. Furthermore, we introduce a (deterministic) volatility modulated pure-jump L\'{e}vy component in the price dynamics.
We start from their setting and investigate how to specify a flexible, yet sufficiently tractable, multivariate additive mean-reverting model. 
We emphasize that our dynamics are described under the empirical probability measure $\bP$.

Let us first introduce the stochastic basis underlying our framework.
\begin{assumption}\label{stoch_basis}
Let $(\Omega,\mathcal{F},\mathbb{F},\bP)$ be a filtered probability space satisfying the usual hypotheses (see \cite{KaS:91}). Let $W$ be a multidimensional Brownian motion and $\overline N(dt,dy):=N(dt,dy)-dt\, \nu(dy)$ denote the compensated Poisson random measure associated to a centered square-integrable pure-jump L\'{e}vy process
\begin{equation}
J(t)=\int_0^t \int_\bR y\,\overline N(ds,dy)
\end{equation}
where the L\'{e}vy measure $\nu(dy)$ is assumed square-integrable in the sense that $\int_{\bR^k} \|y\|^2\,\nu(dy)<\infty$.\footnote{This assumption can be relaxed to $\int_{\bR^k} (\|y\|^2\wedge1)\,\nu(dy)<\infty$ and we could derive analogous results in this section by defining a non-compensated Lévy process $\widetilde J(t):=\int_0^t \int_{\|y\|<1} y\,\overline N(ds,dy)+\int_0^t \int_{\|y\|\geq1} y\,N(ds,dy)$. However, since we want to ease the mathematical discussion and to focus on the modeling intepretation,
we assume the stronger condition $\int_{\bR^k} \|y\|^2\,\nu(dy)<\infty$, that in particular implies that $J$ (as defined in \eqref{stoch_basis}) is a square-integrable martingale component of the SDEs in \eqref{13} and \eqref{14}. Let us remark that for some of the results treated in this paper, we will need even stronger assumptions on $\nu$ (cf. Section 3.2), which makes it somehow unprofitable to start here with more general assumptions.} The processes $W$ and $J$ are both of dimension $k$ and independent of each other. 
\end{assumption}

Assume that we are in a market with $n$ commodities. 
Fix a time horizon $\mathbb{T}$ for our economy and define the sets $\mathcal{A}^{\mathbb{T}}_1=\{ (t,T)\in[0,\mathbb{T}]^2: t\leq T\}$ and $\mathcal{A}^{\mathbb{T}}_2=\{ (t,T_1,T_2)\in[0,\mathbb{T}]^3: t\leq T_1< T_2\}$. For each $(t,T)\in \mathcal{A}^{\mathbb{T}}_1$, let $f(t,T)$ denote the $n$-dimensional vector of forward prices at time $t$ for $n$  forward contracts with instantaneous delivery at time $T$. Analogously, for any $(t,T_1,T_2)\in \mathcal{A}^{\mathbb{T}}_2$, let $F(t,T_1,T_2)$ be the $n$-dimensional vector of swap prices at time $t$ of the $n$ corresponding contracts with delivery period $[T_1,T_2]$. 
Let $c:\mathcal{A}^{\mathbb{T}}_1\to\bR^n$, $\lambda:\mathcal{A}^{\mathbb{T}}_1\to\bR^{n\times n}$, $\sigma:\mathcal{A}^{\mathbb{T}}_1\to\bR^{n\times k}$, $\psi:\mathcal{A}^{\mathbb{T}}_1\to\bR^{n\times k}$, $C:\mathcal{A}^{\mathbb{T}}_2\to\bR^n$, $\Lambda:\mathcal{A}^{\mathbb{T}}_2\to\bR^{n\times n}$, $\Sigma:\mathcal{A}^{\mathbb{T}}_2\to\bR^{n\times k}$, $\Psi:\mathcal{A}^{\mathbb{T}}_2\to\bR^{n\times k}$ be deterministic, measurable vector/matrix fields satisfying the following technical assumptions.
\begin{assumption}\label{1}
\begin{itemize}
\item For any $T\in[0,\mathbb{T}]$, $t\mapsto c(t,T)$ is integrable on $[0,T]$, i.e. $\int_0^{T} \|c(t,T)\| dt$ is finite.
\item For any $T_1,T_2\in[0,\mathbb{T}]$ such that $T_1< T_2$, $t\mapsto C(t,T_1,T_2)$ is integrable on $[0,T_1]$, i.e. $\int_0^{T_1} \|C(t,T_1,T_2)\| dt$ is finite.
\item For each $T\in[0,\mathbb{T}]$, $t\mapsto\sigma(t,T)$ and $t\mapsto\psi(t,T)$ are square integrable on $[0,T]$, i.e. $\int_0^{T} \|\sigma(t,T)\|^2 dt$ and $\int_0^{T} \|\psi(t,T)\|^2 dt$ are finite. 
\item For any $T_1,T_2\in[0,\mathbb{T}]$ such that $T_1< T_2$, $t\mapsto\Sigma(t,T_1,T_2)$ and $t\mapsto\Psi(t,T_1,T_2)$ are square integrable on $[0,T_1]$, i.e. $\int_0^{T_1} \|\Sigma(t,T_1,T_2)\|^2 dt$ and $\int_0^{T_1} \|\Psi(t,T_1,T_2)\|^2 dt$ are finite. 
\item For each $T\in[0,\mathbb{T}]$, the matrix field $t\mapsto\lambda(t,T)$ is continuous on $[0,T]$.
\item For any $T_1,T_2\in[0,\mathbb{T}]$ such that $T_1< T_2$, the matrix field $t\mapsto\Lambda(t,T_1,T_2)$  is continuous on $[0,T_1]$.
\end{itemize}
\end{assumption}
We assume the following dynamics:
\begin{align}\label{13}
df(t,T)&=(c(t,T)-\lambda(t,T) f(t,T))\, dt+\sigma(t,T)\,dW(t)+\psi(t,T)\,dJ(t),
 \\
\nonumber
dF(t,T_1,T_2)&=(C(t,T_1,T_2)-\Lambda(t,T_1,T_2) F(t,T_1,T_2))\, dt\\  \label{14} 
&\qquad\qquad\qquad\quad+\Sigma(t,T_1,T_2)\,dW(t)+\Psi(t,T_1,T_2)\,dJ(t).
\end{align}  
The initial conditions $f(0,T)$ and $F(0,T_1,T_2)$ are deterministic Borel measurable functions in
$T\leq \mathbb T$ and $T_1< T_2\leq \mathbb T$, respectively.

\begin{remark}\label{64}
We would like to have an explicit representation for the solutions of \eqref{13} and \eqref{14}. To simplify the upcoming discussion, consider the $T$-independent continuous version of \eqref{13}: 
\begin{equation}\label{16}
df(t)=-\lambda(t) f(t)\, dt+\sigma(t)\,dW(t),
\end{equation}
where $f(0)$ is the identity matrix, $c\equiv0$ and $\lambda,\sigma$ are assumed sufficiently regular so that there exists a unique solution $f$. 
Let us introduce the following definition.
\begin{definition}[Commutative property]\label{17}
We say that a square matrix-valued continuous function $A(t)$ satisfies the commutative property (CP) if, for any pair of time values $t_1$ and $t_2$, $A(t_1)A(t_2)=A(t_2)A(t_1)$. 
\end{definition} 
If $\lambda(t)$ satisfies (CP), the unique solution $U$ of the matrix ODE
$$
dU(t)=\lambda(t) U(t)\, dt,
$$
can be written as 
\begin{equation*}
U(t)= e^{\int_0^t \lambda(u) du}
\end{equation*}
by introducing the matrix exponential $e^A:=\sum_{k=0}^\infty \frac{A^k}{k!}$ (see e.g. \cite{Blanes}). In particular, since $U(t)$ is invertible and its inverse satisfies
\begin{equation}\label{u_inv}
dU^{-1}(t)=-\lambda(t) U^{-1}(t)\, dt,
\end{equation}
we also have that $\lambda(t)U(t)=U(t)\lambda(t)$.\footnote{This property follows from \eqref{u_inv}, by observing that $0=d(U(t)U^{-1}(t)) = dU(t)U^{-1}(t)+U(t)dU^{-1}(t)=\lambda(t) U(t)U^{-1}(t)-U(t)\lambda(t) U^{-1}(t)=\lambda(t)-U(t)\lambda(t) U^{-1}(t)$.} These results do not hold in general when $\lambda$ does not satisfy (CP), due to the non-commutativity of the matrix product. Now, apply It\^{o}'s Lemma to $U(t)f(t)$:
\begin{align}
d\left(U(t)f(t)\right) &= \lambda(t)\,U(t)\,f(t)\,dt + U(t)\,df(t)\\
&= \left(\lambda(t)\,U(t) - U(t)\,\lambda(t)\right) f(t)\, dt + U(t)\,\sigma(t)\,dW(t)\\
&= U(t)\,\sigma(t)\,dW(t).
\end{align}
By integrating the above equality, we derive that if $\lambda(t)$ satisfies (CP), then $f(t)$ can be written explicitly as 
\begin{equation}\label{15}
f(t)= e^{-\int_0^t \lambda(u) du}\,\int_0^t e^{\int_0^s \lambda(u) du} \sigma(s) dW(s).
\end{equation}
\end{remark} 

In view of Remark \ref{64}, we introduce the following assumption.
\begin{assumption}\label{cp}
The matrix-valued functions $\lambda(\cdot,T)$ and $\Lambda(\cdot,T_1,T_2)$ satisfy (CP) for all $T\leq\bT$ and $T_1<T_2\leq\bT$.
\end{assumption}
In particular, let us remark that (CP) is fulfilled by any matrix constant in time and any time-dependent diagonal matrix. Assumption \ref{1} ensures that \eqref{13} and \eqref{14} admit a unique square-integrable solution (cf. also \cite[Appendix A]{Benth20081116}). Moreover, Assumption \ref{cp} (see Remark \ref{64}) guarantees that the solution of \eqref{13} can be expressed explicitly as
\begin{eqnarray}\nonumber
f(t,T) & = & e^{-\int_0^t \lambda(s,T) ds} f(0,T)+\int_0^t e^{-\int_s^t \lambda(u,T) du} c(s,T) ds\\ \label{95}
& & +\int_0^t e^{-\int_s^t \lambda(u,T) du} \sigma(s,T) dW(s)
+\int_0^t e^{-\int_s^t \lambda(u,T) du} \psi(s,T) dJ(s)
\nonumber
\end{eqnarray}
and, analogously, the solution of \eqref{14} can be written as
\begin{eqnarray}\nonumber
F(t,T_1,T_2) & = & e^{-\int_0^t \Lambda(s,T_1,T_2) ds} F(0,T_1,T_2)+\int_0^t e^{-\int_s^t \Lambda(u,T_1,T_2) du} C(s,T_1,T_2) ds \label{95F}\\
&   & +\int_0^t e^{-\int_s^t \Lambda(u,T_1,T_2) du} \Sigma(s,T_1,T_2) dW(s)
+\int_0^t e^{-\int_s^t \Lambda(u,T_1,T_2) du} \Psi(s,T_1,T_2) dJ(s). \nonumber
\end{eqnarray}

These processes are also linked among themselves by no-arbitrage. In fact, 
in this continuous time setting, a limit argument (see, for instance, \cite{bjerksund2010valuation}) leads to the following no-arbitrage condition among each swap and the corresponding family of forwards: for any $t\leq T_1$ and $T_2>T_1>0$ 
\begin{equation}\label{2}
F(t,T_1,T_2)=\int_{T_1}^{T_2} \widehat{w}(T,T_1,T_2) f(t,T)\,dT,
\end{equation}
where $\widehat w:\mathcal{A}^{\mathbb{T}}_2\to\bR^{n\times n}$ is a weight function and represents the time value of money. This can differ from contract to contract according to how settlement takes place (see Equation 4.2 in \cite{Benth20081116}). For example, in the one-dimensional case we could have $\widehat{w}(T,T_1,T_2)=\frac1{T_2-T_1}$ or, by assuming a continuously compounded constant risk-free rate $r$, $\widehat{w}(T,T_1,T_2)=\frac{e^{-r T}}{\int_{T_1}^{T_2} e^{-r T}\,dT}$. Thus, if we generalize in a straightforward way to a multidimensional market, then we could define, for instance,  
$$
\widehat{w}(T,T_1,T_2)=\frac1{T_2-T_1}\, I_n\qquad
\mbox{or}\qquad \widehat{w}(T,T_1,T_2)=\frac{e^{-r T}}{\int_{T_1}^{T_2} e^{-r T}\,dT}\, I_n,
$$
with $I_n$ being the identity matrix of dimension $n$. 
We may further enrich this setting by interpreting in a wider sense the role of one, or more, components of $f(t,T)$ and, consequently, the role of $\widehat{w}$ as well. For example, the last component of $f(t,T)$ may represent some \emph{forward} stochastic factor with dynamics governed by \eqref{13}, which affects the value of the swap contracts  $F(t,T_1,T_2)$ through \eqref{2}. 
However, we do not want to look into this matter at the present time, leaving the possibility to investigate it to future work. 
In view of its modeling intepretation, we define the weight function in the following way.
\begin{definition}\label{65}
A matrix-valued map $\widehat w:\mathcal{A}^{\mathbb{T}}_2\to\bR^{n\times n}$ is called a weight function if, for any $T_1<T_2$, it is integrable with respect to $T$ on $[T_1,T_2]$  and
$$
\int_{T_1}^{T_2} \widehat{w}(T,T_1,T_2)\,dT=I_n.
$$
\end{definition}

In the next proposition we state how the no-arbitrage condition in \eqref{2} determines the coefficients appearing in the dynamics \eqref{13} and \eqref{14}. Let us observe that, in order to do this, we assume that the mean-reversion speeds $\lambda$ and $\Lambda$ are independent of delivery. 
\begin{prop} \label{52}
For all $T\leq\bT$ and $T_1<T_2\leq\bT$, let $f(\cdot,T)$ and $F(\cdot,T_1,T_2)$ be the unique solutions of \eqref{13} and \eqref{14} for given coefficients 
satisfying Assumption \ref{1}, where, in addition, the mean-reversion coefficients are independent of delivery: 
\begin{equation}\label{20}
\Lambda(t,T_1,T_2)=\lambda(t,T):=\lambda(t).
\end{equation}
Let us assume that the no-arbitrage relation \eqref{2} holds for a given weight function $\widehat w$ satisfying the following integrability conditions:
for all $T_1 < T_2$, 
\begin{align*}
&\int_0^{T_1} \int_{T_1}^{T_2} \|\widehat{w}(T,T_1,T_2)\,\sigma(t,T)\|^2\,dT\,dt<\infty,\, 
\int_0^{T_1} \int_{T_1}^{T_2} \|\widehat{w}(T,T_1,T_2)\,\psi(t,T)\|^2\,dT\,dt<\infty, \\
&\int_0^{T_1} \int_{T_1}^{T_2} \|\widehat{w}(T,T_1,T_2)\,c(t,T)\|\,dT\,dt<\infty,\, 
\int_{T_1}^{T_2} \|\widehat{w}(T,T_1,T_2)\,f(0,T)\|\,dT<\infty, \\
&\int_0^{T_1} \int_{T_1}^{T_2} \|\widehat{w}(T,T_1,T_2)\,\lambda(t)\,f(t,T)\|^2\,dT\,dt<\infty, \quad\mbox{a.s.}
\end{align*}
If the matrix $\widehat w(T,T_1,T_2)$ commutes with $\lambda(t)$ for every $T\leq T_1<T_2$ and $t<T$, then
\begin{align}
\label{19}
C(t,T_1,T_2)&=\int_{T_1}^{T_2} \widehat w(T,T_1,T_2) c(t,T)\,dT, \\ \label{21}
\Sigma(t,T_1,T_2)&=\int_{T_1}^{T_2} \widehat w(T,T_1,T_2) \sigma(t,T)\,dT,\\
\label{192021}
\Psi(t,T_1,T_2)&=\int_{T_1}^{T_2} \widehat w(T,T_1,T_2) \psi(t,T)\,dT.
\end{align}
\end{prop}
\begin{proof}
The proof is an application of the stochastic Fubini Theorem (cf. e.g.  \cite[Theorem 64]{MR2020294}) and follows by comparing the coefficients of \eqref{13} to \eqref{14} after integrating against $\widehat w$. From \eqref{13} we can write
\begin{align}\nonumber
\int_{T_1}^{T_2} &\widehat{w}(T,T_1,T_2) f(t,T)\,dT= \int_{T_1}^{T_2} \widehat{w}(T,T_1,T_2)f(0,T)\,dT\\ \nonumber
&+\int_{T_1}^{T_2}\int_0^t \widehat{w}(T,T_1,T_2)  (c(u,T)-\lambda(u) f(u,T))\,du \,dT\\ \nonumber
&+\int_{T_1}^{T_2} \int_0^t\widehat{w}(T,T_1,T_2)  \sigma(u,T)\, dW(u)\,dT
+\int_{T_1}^{T_2} \int_0^t  \widehat{w}(T,T_1,T_2) \psi(u,T)\, dJ(u)\,dT.
\end{align}
After observing that
$$
\int_0^t \widehat{w}(T,T_1,T_2) \lambda(u) f(u,T)\,du =\int_0^t \lambda(u)\widehat{w}(T,T_1,T_2) f(u,T)\,du,
$$
we change the order of integration in the last equation (recalling \eqref{2}) so to obtain
\begin{align}\nonumber
F&(t,T_1,T_2)= F(0,T_1,T_2)+\int_0^t \left(\int_{T_1}^{T_2} \widehat{w}(T,T_1,T_2)  c(u,T)\,dT - \lambda(u) F(u,T_1,T_2)\right)\,du\\ \nonumber
&+ \int_0^t \int_{T_1}^{T_2} \widehat{w}(T,T_1,T_2)  \sigma(u,T)\,dT\, dW(u)
+\int_0^t \int_{T_1}^{T_2} \widehat{w}(T,T_1,T_2) \psi(u,T)\,dT\, dJ(u).
\end{align}
Then, the result follows from the uniqueness of representation of solutions by comparing the coefficients of the last equation to the respective ones in \eqref{14}.
\end{proof}

We finish this section by noticing that, under sufficient regularity assumptions on the evolution of forward prices, we can write down the implied dynamics of the spot price, defined as $S(t)=f(t,t)$.
\begin{prop} In addition to Assumption \ref{1}, let us assume that $c(t,T),\sigma(t,T),\psi(t,T)$ and $f(t,T)$ are differentiable with respect to $T$ for any $t\leq T$, with bounded partial derivatives $c_T(t,T),\sigma_T(t,T),\psi_T(t,T)$ and $f_T(t,T)$.
Then, the spot price follows the dynamics
$$
dS(t)=(c(t,t)+\zeta(t)-\lambda(t) S(t))\, dt +\sigma(t,t)dW(t) +\psi(t,t)dJ(t),
$$
where 
$$
\zeta(t):=f_T(0,t)+\int_0^t (c_T(u,t)-\lambda(u) f_T(u,t))\,du+\int_0^t \sigma_T(u,t)\, dW(u)+\int_0^t \psi_T(u,t)\, dJ(u).
$$
\end{prop}
\begin{proof}
The proof is an application of the stochastic Fubini Theorem and follows the same steps of Proposition 11.1.1 in \cite{MR1474500}.
\end{proof}
Naturally, when starting out with a mean-reverting dynamics for swaps and forwards, the implied spot price dynamics will also become mean-reverting. Indeed, we observe that spot prices follow a multidimensional Lévy
Ornstein-Uhlenbeck process with time-dependent speed of mean-reversion $\lambda(t)$ and volatility
$\sigma(t,t)$. It mean-reverts towards a time-dependent level $\lambda(t)^{-1}(c(t,t)+\zeta(t))$ (whenever $\lambda(t)$ is invertible). This opens for including
seasonal variations into the model, which is very relevant in energy markets.

\section{A Mean-Reverting Model consistent with No-Arbitrage}

Now we describe how to construct, under a suitable assumption, a HJM-type forward market as in the previous section. We state how this naturally leads us to characterize the model, firstly in terms of its dynamical behavior, and secondly in relation to the existence of equivalent martingale measures and, therefore, arbitrage. In this regard, we show two crucial results about the martingale property of stochastic exponentials. The main result (Theorem \ref{85}) allows us 
to validate general measure changes with Esscher-type jump component and stochastic kernel in the diffusive part. Then, a different approach is proposed for the specific case of stochastic exponential of Brownian integrals of continuous kernels, so that, in particular, the density process is continuous. Finally, we discuss why this technique does not apply to the case of general Lévy kernels and how this fact is related to asymptotic properties of even moments of infinitely divisible distributions.

Let us begin with our main modeling assumption. We consider the same stochastic basis as in Assumption \ref{stoch_basis}.
\begin{assumption}\label{18}
Take an integer $m\in\bN$. Assume that the following stochastic differential equation 
\begin{equation}\label{22}
dX(t)=b(t,X(t))dt+\nu(t,X(t))dW(t)+\eta(t,X(t-))dJ(t)
\end{equation}
admits a unique solution $X$ taking values in $\bR^{m}$, for a drift $b:[0,\mathbb{T}]\times\bR^{m}\to \bR^{m}$ and volatility coefficients $\nu,\eta:[0,\mathbb{T}]\times\bR^{m}\to \bR^{m\times k}$. The forward prices $f(t,T)$ have an affine representation
\begin{equation}\label{4}
f(t,T)=\alpha(t,T)X(t)+\beta(t,T),
\end{equation}
for some deterministic matrix/vector-valued functions $\alpha:\mathcal{A}^{\mathbb{T}}_1\to\bR^{n\times m}$, $\beta:\mathcal{A}^{\mathbb{T}}_1\to\bR^{n}$, which we assume to be bounded on $\mathcal{A}^{\mathbb{T}}_1$ and
continuously differentiable in $t\in[0,T]$, for all $T\leq\bT$.
In view of the no-arbitrage condition \eqref{2}, we define the swap price process $F(t,T_1,T_2)$ by
\begin{equation}\label{112}
F(t,T_1,T_2)=\int_{T_1}^{T_2} \widehat{w}(T,T_1,T_2) f(t,T)\,dT,
\end{equation}
where $\widehat w$ is a weight function as in Definition \ref{65}.
\end{assumption}

The process $X$ can be interpreted as a stochastic state variable underlying the forward curves for \emph{all} maturities $T$. This might be related to the spot price, which is the case, for instance, of the models that we propose in Sections 4 and 5. Observe, in particular, that the stochastic processes $f(\cdot,T)$ and $F(\cdot,T_1,T_2)$ are Markovian. The question of Markovianity for this class of models has been discussed also in \cite[Equation 4.8]{Benth20081116}, where the authors show that even simple log-normal specifications of the forward dynamics lead, in general, to non-Markovian swap price processes (unless interpreted as infinite-dimensional stochastic processes). Instead, Assumption \ref{18} allows us to preserve the Markov property of our models and, thus, analytical tractability.

\subsection{Dynamics}

Since we want a dynamical behavior of mean-reverting type for $f$ and $F$ (see \eqref{13} and \eqref{14}), the affine structure of forward prices assumed in \eqref{4} determines in a rather natural way the corresponding functional form of the coefficients expected in \eqref{22}. 
In this regard, we have two symmetrical results. 

\begin{prop}\label{27}
Assume in \eqref{22} that the coefficients $b,\nu$ and $\eta$ take the following affine form
\begin{align}\label{31}
b(t,X(t))&=\theta(t)+\Theta(t)X(t),\\ \label{31b}
\nu(t,X(t))&=v(t),\\ \label{31c}
\eta(t,X(t))&=z(t),
\end{align}
where $\theta:[0,\mathbb{T}]\to\bR^{m}$ is integrable, $v:[0,\mathbb{T}]\to\bR^{m\times k}$ and $z:[0,\mathbb{T}]\to\bR^{m\times k}$ are square-integrable and $\Theta:[0,\mathbb{T}]\to\bR^{m\times m}$ is bounded. If there exists a continuous matrix-valued function $\lambda:[0,\mathbb{T}]\to\bR^{n\times n}$ such that
\begin{align}\label{23}
\lambda(t)\alpha(t,T)&=-\alpha_t(t,T)-\alpha(t,T)\Theta(t),
\end{align}
then, by defining
\begin{align}
\label{24}
c(t,T)&:=\lambda(t)\beta(t,T)+\beta_t(t,T)+\alpha(t,T)\theta(t),
\\ \label{25}
\sigma(t,T)&:=\alpha(t,T)v(t),\\
\label{91}
\psi(t,T)&:=\alpha(t,T)z(t),
\end{align}
the stochastic process $f(t,T)$ is the unique solution of
\begin{equation*}\label{26}
df(t,T)=(c(t,T)-\lambda(t) f(t,T))\, dt+\sigma(t,T)dW(t)+\psi(t,T)dJ(t).
\end{equation*}
Furthermore, if $\widehat w$ satisfies the assumptions of Proposition \ref{52}, $F(t,T_1,T_2)$ is the unique solution of
\begin{equation*}
dF(t,T_1,T_2)=(C(t,T_1,T_2)-\lambda(t) F(t,T_1,T_2))\, dt+\Sigma(t,T_1,T_2)\,dW(t)+\Psi(t,T_1,T_2)\,dJ(t),
\end{equation*}
where $C,\Sigma, \Psi$ are defined in \eqref{19},\eqref{21}, \eqref{192021}.
\end{prop}
\begin{proof}
Let us apply It\^o's Lemma to
$$
f(t,T)=\alpha(t,T)X(t)+\beta(t,T).
$$
Hence,
\begin{align}\nonumber
df(t,T)=(\beta_t(t,T)+\alpha(t,T)\theta(t)+&\alpha_t(t,T)X(t)+\alpha(t,T)\Theta(t)X(t))dt\\ \label{108}
&+\alpha(t,T)v(t)dW(t)+\alpha(t,T)z(t)dJ(t).
\end{align}
Since, by assumption, there exists a function $\lambda:[0,\mathbb{T}]\to\bR^{n\times n}$ satisfying
\begin{equation*}
\alpha_t(t,T)+\alpha(t,T)\Theta(t)=-\lambda(t)\alpha(t,T)
\end{equation*}
whereas $c(t,T)$, $\sigma(t,T)$ and $\psi(t,T)$ satisfy, by definition,
\begin{align*}
\beta_t(t,T)+\alpha(t,T)\theta(t)&=c(t,T)-\lambda(t)\beta(t,T),
\\
\alpha(t,T)v(t)&=\sigma(t,T),\\
\alpha(t,T)z(t)&=\psi(t,T),
\end{align*}
the statement for $f$ follows after substituting these expressions into \eqref{108}. The result for $F$ can be proven along the same lines of Proposition \ref{52}.
\end{proof}
The following proposition can be seen as the converse of the previous one.
\begin{prop}\label{58}
If $f(t,T)$ is the unique solution of 
\begin{equation}\label{76}
df(t,T)=(c(t,T)-\lambda(t) f(t,T))\, dt+\sigma(t,T)dW(t)+\psi(t,T)dJ(t),
\end{equation}
for $c,\sigma,\psi$ and $\lambda$ satisfying Assumption \ref{1}, then $b,\nu$ and $\eta$ in \eqref{22} take the following affine form
\begin{align}\label{100}
\alpha(t,T)b(t,X(t))&=\widetilde\theta(t,T)+\widetilde\Theta(t,T) X(t)\\ \label{101}
\alpha(t,T)\nu(t,X(t))&=\sigma(t,T),\\ \label{102}
\alpha(t,T)\eta(t,X(t))&=\psi(t,T),
\end{align}
with
\begin{align}\label{28}
\widetilde\theta(t,T)&=c(t,T)-\lambda(t)\beta(t,T)-\beta_t(t,T),\\ \label{29}
\widetilde\Theta(t,T)&=-\lambda(t)\alpha(t,T)-\alpha_t(t,T).
\end{align}
\end{prop} 
\begin{proof}
As in the proof of Proposition \ref{27}, the statement is direct consequence of It\^{o}'s Lemma applied to
$$
f(t,T)=\alpha(t,T)X(t)+\beta(t,T),
$$ 
which gives that
\begin{align*}
df(t,T)=(\beta_t(t,T)+\alpha_t(t,T)X(t)&+\alpha(t,T)b(t,X(t)))dt\\
&+\alpha(t,T)\nu(t,X(t))dW(t)+\alpha(t,T)\eta(t,X(t-))dJ(t).
\end{align*}
It is then sufficient to compare the coefficients of the last equation to the ones of \eqref{76}. 
\end{proof}

\subsection{Arbitrage}

We now investigate the question of arbitrage in the context of our HJM-type forward model. 
A sufficient condition for an arbitrage-free market, together with \eqref{112}, is the existence of an equivalent martingale measure, which is a probability measure $\bQ$ equivalent to $\bP$ such that the discounted price processes of all traded contracts are $\bQ$-martingales.
Notice that forwards and swaps are costless to enter, so their price processes should be $\bQ$-martingale without any discounting. 
We also mention that the well-known \emph{HJM drift condition} for forward {\it rate} models in interest rate theory is not 
present in our context, as forward contracts are tradeable assets unlike the forward rates, which model in a nonlinear way the bond price dynamics. 

Firstly, we introduce the candidate density processes. 

\begin{assumption}\label{105}
We assume that $N(ds,dy)$ is the Poisson random measure of a Lévy process $J$ with independent components $J_j$, for $j=1,\ldots,k$. We indicate by $N_j(ds,dy_j)$ the Poisson random measure of $J_j$ and by $\overline N_j(ds,dy_j)  = N_j(ds,dy_j)-ds\,\nu_j(dy_j)$ the corresponding compensated version. 
\end{assumption}
Let $\phi=(\phi^{(j)})_{j=1,\ldots,k}$ and $\xi=(\xi^{(j)})_{j=1,\ldots,k}$ be $k$-dimensional adapted processes such that
\begin{equation}\label{109}
\bE\left[\int_0^{\mathbb{T}} \|\phi(s)\|^2 ds\right]<\infty,\qquad 
\bE\left[\int_0^{\mathbb{T}} \|\xi(s)\|^2 ds\right]<\infty,
\end{equation}
and define, for $t\in[0,\mathbb{T}]$, the process $Z$ as the unique strong solution of 
\begin{equation}\label{90}
dZ(t)=Z(t-)\,dH(t),
\end{equation}
such that $Z(0)=1$, where
\begin{align}\nonumber
H(t):=&\int_0^t \phi^\intercal(s) dW(s)+\int_0^t \int_{\bR^k}\xi^\intercal(s-) \,y\,\overline N(ds,dy),\\ \label{87}
&=\sum_{j=1}^k \left(\int_0^t \phi^{(j)}(s) dW_j(s)+\int_0^t \int_{\bR}\xi^{(j)}(s-) \,y_j\,\overline N_j(ds,dy_j)\right).
\end{align}
If the processes $\phi$ and $\xi$ satisfy \eqref{109}, $H$ is a well-defined square-integrable martingale.
The process $Z$ is called the \emph{stochastic} or \emph{Dol\'{e}ans-Dade exponential} of $H$ and is sometimes indicated as
$$
Z(t)=\mathcal{E}(H)(t).
$$
More explicitly, it can be written as (cf. for instance \cite[Theorem II.37]{MR2020294})
\begin{align}\label{86}
Z(t)=e^{H(t)-\frac12 \int_0^t \|\phi(s)\|^2\ ds}\prod_{0<s\leq t} (1+\Delta H(s))\,e^{-\Delta H(s)}.
\end{align} 
If $Z$ is a strictly positive martingale, we can introduce an equivalent probability measure $\bQ$ defining its Radon-Nikodym derivative as
\begin{equation}\label{70}
\frac{d\bQ}{d\bP}:=Z(\mathbb{T}).
\end{equation}
Define the stochastic process 
\begin{equation}\label{wq}
W^\bQ(t):=W(t)-\int_0^t\phi(s)\,ds,
\end{equation}
the random measure 
\begin{equation}\label{nq}
\overline N^\bQ(dt,dy):=\overline N(dt,dy)-\xi(t)^\intercal y\,\nu(dy)dt 
\end{equation}
and the jump process
\begin{equation}\label{jq}
J^\bQ(t):=\int_{\bR^k} y\,\overline N^\bQ(dt,dy)=J(t)-\int_0^t K\,\xi(s) ds, 
\end{equation}
where
\begin{equation*}
K=\mathrm{diag}(\kappa),\qquad \kappa=(\kappa_j)_{j=1,\dots,k},\qquad \kappa_j=\int_\bR y^2\,\nu_j(dy).
\end{equation*}

The next two propositions identify the subfamily of models 
admitting the existence of an equivalent martingale measure. 
In the remainder of this paper, we assume the same assumptions of 
Proposition \ref{27}, so that, in particular, $X$ evolves as 
\begin{equation}\label{89}
dX(t)=\bigl(\theta(t)+\Theta(t)X(t)\bigr)dt+v(t)dW(t)+z(t)dJ(t).
\end{equation}


\begin{prop}\label{61}
If $f$ satisfies 
$$
df(t,T)=\sigma(t,T)\,dW^\bQ(t)+\psi(t,T)\,dJ^\bQ(t),
$$
then 
\begin{equation}\label{67}
\alpha(t,T)\bigl(v(t)\phi(t)+z(t)\,K\,\xi(t)\bigr)=\Phi_1(t,T)X(t)+\Phi_0(t,T),
\end{equation}
for $\Phi_0$ and $\Phi_1$ denoting
\begin{align}\label{32}
\Phi_1(t,T)&=-\alpha_t(t,T)-\alpha(t,T)\Theta(t),\\ \label{33} 
\Phi_0(t,T)&=-\beta_t(t,T)-\alpha(t,T)\theta(t).
\end{align}
In particular, if
\begin{align}\label{83}
\phi(t)&=\phi_1(t)X(t)+\phi_0(t),\\ \label{84}
\xi(t)&=\xi_1(t)X(t)+\xi_0(t),
\end{align}
then
\begin{align}\label{81}
\Phi_1(t,T)&=\alpha(t,T)\bigl(v(t)\phi_1(t)+z(t)\,K\,\xi_1(t)\bigr),\\ \label{82} 
\Phi_0(t,T)&=\alpha(t,T)\bigl(v(t)\phi_0(t)+z(t)\,K\,\xi_0(t)\bigr).
\end{align}
Finally, $\alpha$ and $\beta$ satisfy the following linear ordinary differential equations:
\begin{align}\label{6}
\alpha_t(t,T)=-\alpha(t,T)\gamma_1(t),\\ \label{8} 
\beta_t(t,T)=-\alpha(t,T)\gamma_0(t),
\end{align}
where $\gamma_0$ and $\gamma_1$ are
\begin{align}\label{9}
\gamma_1(t)&=v(t)\phi_1(t)+z(t)\,K\,\xi_1(t)+\Theta(t),\\ \label{10} 
\gamma_0(t)&=v(t)\phi_0(t)+z(t)\,K\,\xi_0(t)+\theta(t).
\end{align}
\end{prop} 
\begin{proof}
It\^{o}'s Lemma applied to $f(t,T)=\alpha(t,T)X(t)+\beta(t,T)$ gives
\begin{align*}
df(t,T)=(\beta_t(t,T)+\alpha(t,T)\theta(t)+\alpha_t(t,T)X(t)&+\alpha(t,T)\Theta(t)X(t))dt\\
&+\alpha(t,T)v(t)dW(t)+\alpha(t,T)z(t)dJ(t).
\end{align*}
By replacing $W$ and $J$ with $W^\bQ$ and $J^\bQ$, we get 
\begin{align*}
df(t,T)=\bigl(&\beta_t(t,T)+\alpha(t,T)\theta(t)+\alpha_t(t,T)X(t)+\alpha(t,T)\Theta(t)X(t)
\\
&+\alpha(t,T)v(t)\phi(t)+\alpha(t,T)z(t)\,K\,\xi(t)\bigr)dt+\alpha(t,T)z(t)dJ^\bQ(t)+\alpha(t,T)v(t)dW^\bQ(t).
\end{align*}
In order to have zero drift, \eqref{67} must hold for $\Phi_0$ and $\Phi_1$ as in \eqref{32} and \eqref{33}. In particular, the last equation holds if $\phi$ and $\xi$ take the affine form \eqref{83} and \eqref{84} with $\Phi_0$ and $\Phi_1$ defined by \eqref{81} and \eqref{82}, so that
\begin{align*}
\alpha_t(t,T)+\alpha(t,T)\Theta(t)+\alpha(t,T)v(t)\phi_1(t)+\alpha(t,T)z(t)\,K\,\xi_1(t)&=0,\\ 
\beta_t(t,T)+\alpha(t,T)\theta(t)+\alpha(t,T)v(t)\phi_0(t)+\alpha(t,T)z(t)\,K\,\xi_0(t)&=0,
\end{align*}
which yield \eqref{6} and \eqref{8}.
\end{proof}
The converse result holds as well.
\begin{prop}\label{60}
Assume 
that the following conditions are fulfilled:
\begin{itemize}
\item $\alpha$ and $\beta$ satisfy, for some continuous functions $\gamma_0:[0,\mathbb{T}]\to\bR^m$ and $\gamma_1:[0,\mathbb{T}]\to\bR^{m\times m}$, the linear ordinary differential equations \eqref{6} and \eqref{8};

\item  there exist adapted processes $\phi(t)$ and $\xi(t)$ satisfying \eqref{109} of the form
\begin{align*}
\phi(t)&=\phi_1(t)X(t)+\phi_0(t),\\
\xi(t)&=\xi_1(t)X(t)+\xi_0(t),
\end{align*}
where $\phi_0,\phi_1,\xi_0$ and $\xi_1$ are in relation with $\gamma_0$ and $\gamma_1$ by \eqref{9} and \eqref{10};

\item the process $Z$ defined in \eqref{90} is a strictly positive $\bP$-martingale on $[0,\mathbb{T}]$.
\end{itemize}
Then, by defining $\frac{d\bQ}{d\bP}=Z(\bT)$, it holds that $W^\bQ$ in \eqref{wq} is a Brownian motion under $\bQ$ and $\overline N^\bQ(dt,dy)$ defined in \eqref{nq} is the $\bQ$-compensated Poisson random measure of $N$ in the sense of \cite[Theorem 1.35]{MR2322248}. Also, we have
\begin{align*}
df(t,T)&=\sigma(t,T)\,dW^\bQ(t)+\psi(t,T)\,dJ^\bQ(t),
\\
dF(t,T_1,T_2)&=\Sigma(t,T_1,T_2)\,dW^\bQ(t)+\Psi(t,T_1,T_2)\,dJ^\bQ(t),
\end{align*} 
with the same coefficients as in Proposition \ref{27} and, in particular, the processes $f(\cdot,T)$ and $F(\cdot,T_1,T_2)$ are $\bQ$-martingales for all $T$ and $T_1<T_2$.
\end{prop}
\begin{proof}
By similar arguments as in Proposition \ref{61}, it is enough to apply It\^{o}'s Lemma to $f(t,T)=\alpha(t,T)X(t)+\beta(t,T)$, which gives
$$
df(t,T)=\sigma(t,T)\,dW^\bQ(t)+\psi(t,T)\,dJ^\bQ(t).
$$
Since $F(t,T_1,T_2)=\int_{T_1}^{T_2} \widehat{w}(T,T_1,T_2) f(t,T)\,dT$, we easily obtain that
$$
dF(t,T_1,T_2)=\Sigma(t,T_1,T_2)dW^\bQ(t)+\Psi(t,T_1,T_2)\,dJ^\bQ(t).
$$
As, by assumption, $Z$ is a $\bP$-martingale, the measure $\bQ$ is an equivalent probability measure and we can conclude by Girsanov's Theorem (e.g. \cite[Theorem 1.35]{MR2322248}).
\end{proof}

In Proposition \ref{60} we assumed that $Z$ is a strictly positive $\bP$-martingale on $[0,\mathbb{T}]$, which implies that $Z(\mathbb{T})$ defines the Radon-Nikodym derivative of an equivalent probability measure $\bQ$. As a consequence of Proposition \ref{60}, $\bQ$ is in fact a martingale measure for the market. However, in general $Z$ is only a local martingale. Thus, the interesting question is under which conditions we can prove that $Z$ is a strictly positive true martingale. Tracing through Propositions \ref{61} and \ref{60}, it is natural to have martingale measure kernels in affine form: 
\begin{align}\label{110}
\phi(t)&=\phi_1(t)X(t)+\phi_0(t),\\ \label{111}
\xi(t)&=\xi_1(t)X(t)+\xi_0(t).
\end{align}
The first condition to check is the positivity. From the representation in \eqref{86} we see that $Z$ is strictly positive if and only if $\Delta H > -1$ a.s. 
This boils down to verify that
$$
\Delta H(t) = \xi^\intercal(t-) \Delta J(t) > -1,\qquad\mbox{for every }t\in[0,\bT], \mbox{ a.s.},
$$
which holds under the following assumption.
\begin{assumption}\label{103}
The jump risk parameter $\xi$ is a bounded and deterministic vector function on $[0,\bT]$ (i.e. $\xi_1\equiv0$ in \eqref{111}) such that $\xi^\intercal(t) \,y>-1$ for $\nu$-a.e. $y\in\bR^k$ and each $t\geq0$.
\end{assumption} 
Observe that, in order to define a positive density process, we assume that the jump kernel $\xi$ is deterministic. 

Since a positive local martingale is a supermartingale, in order to prove that $Z$ is a true martingale it is sufficient to verify that $\bE[Z(\bT)]=1$.  Let us then state the main result of this section. 

\begin{theorem}\label{85}
Let us assume that there exist bounded measurable deterministic functions $\phi_0:[0,\mathbb{T}]\to\bR^k$, $\phi_1:[0,\mathbb{T}]\to\bR^{k\times m}$. 
Define the Girsanov kernels $\phi(t):=\phi_1(t)X(t)+\phi_0(t)$ and $\xi$ as in Assumption \ref{103}.
Assume, furtherly, that $\nu$ has fourth moment, that is $\int_{\bR^k} \|y\|^4\,\nu(dy)<\infty$.
Then, the process $Z$ defined by \eqref{90} is a strictly positive true martingale.
\end{theorem}
\begin{proof}
See Appendix A.
\end{proof}


In the next theorem we state the martingale property of the stochastic exponential in the case we have no jumps in the dynamics, that is $\nu\equiv 0$, in an alternative way.
The Novikov condition (\cite{novikov}) is a standard way to prove it in the case of continuous density process, i.e. $H(t)=\int_0^t \phi^\intercal(s) dW(s)$. However, applying the Novikov condition for our state-dependent $\phi$ will yield a valid measure change only for a restricted interval of time $t\in[0,\tau]$, where $\tau$ will depend on the parameters in the model (inspect the proof in Appendix B to see this). Thus, the Novikov condition may fail to validate a martingale measure for all times in question and we cannot conclude that the dynamics are arbitrage-free for $t\geq\tau$. Therefore, we apply a weaker Novikov-type criterion, which relies on asymptotic properties of Gaussian moments. We remark in passing that a nice account of the literature related to this problem can be found in the introduction of \cite{kallsen2002}.

\begin{theorem}\label{7} 
Assume that $\Theta:[0,\mathbb{T}]\to\bR^{m\times m}$ is a bounded measurable function satisfying (CP), $\theta:[0,\mathbb{T}]\to\bR^{m}$ is integrable and $v:[0,\mathbb{T}]\to\bR^{m\times k}$ is square-integrable. Let $\phi_0:[0,\mathbb{T}]\to\bR^k$, $\phi_1:[0,\mathbb{T}]\to\bR^{k\times m}$ be measurable fields such that $\phi_1(t)$ is bounded and $\int_0^{\mathbb{T}} \|\phi_0(s)\|^2 ds<\infty$.
Define the Girsanov kernel
$$
\phi(t)=\phi_1(t)X(t)+\phi_0(t),
$$
where $X$ is the unique solution of
$$
dX(t)=(\theta(t)+\Theta(t)X(t))dt+v(t)dW(t).
$$
Then
\begin{equation}\label{88}
Z(t)=\exp\left(\int_0^t \phi^\intercal(s)dW(s)-\frac12\int_0^t \|\phi(s)\|^2 ds \right)
\end{equation}
is a (strictly positive) $\bP$-martingale on $[0,\mathbb{T}]$. 
\end{theorem}
\begin{proof}
See Appendix B.
\end{proof}

We will now show that the same technique employed in the proof of Theorem \ref{7} fails when adding the jump component in the dynamics of $X$, even if the density process $Z$ is still continuous, i.e. $\xi\equiv0$. We need that the pure-jump Lévy component $J$ satisfies the following property related to the asymptotics of \emph{even} moments of stochastic integrals with respect to L\'{e}vy processes, which does not hold in very simple cases.
\begin{definition} We say that a L\'{e}vy process $J(t):=\int_0^t\int_{\bR^k} y \, \overline{N}(ds,dy)$ satisfies the property $(P)$ if, for any $t\in[0,\mathbb{T}]$ and bounded measurable function $h:[0,\mathbb{T}]\to\bR^{m\times k}$, it holds that
\begin{equation}
\limsup_{n\to\infty} \frac 1n\,\frac{\bE\left[\|\int_0^t\int_{\bR^k} h(s)\, y \, \overline{N}(ds,dy)\|^{2n}\right]}{\bE\left[\|\int_0^t\int_{\bR^k} h(s)\, y \, \overline{N}(ds,dy)\|^{2n-2}\right]}
\leq C,
\end{equation}
where $C=C(t,\mathbb{T},h)$ is a constant possibly depending on $t,\mathbb{T}$ and the function $h$.
\end{definition}
Let us consider, for simplicity, the one dimensional case, i.e. $k=1$, and fix $h\equiv 1$, $t=\mathbb{T}=1$. Then, property (P) becomes
\begin{equation}\label{45}
\limsup_{n\to\infty} \frac 1n\,\frac{\bE\left[J(1)^{2n}\right]}{\bE\left[J(1)^{2n-2}\right]}
\leq C.
\end{equation}
Thus, it boils down to the asymptotic behavior of the moments of infinitely divisible distributions. In general, explicit formulas are available for moments of L\'{e}vy processes in terms of cumulants, but involve rather complicated combinatorial quantities, which do not allow to interpret in a straightforward way the growth behavior with respect to $n$ (see e.g. \cite{Privault2016} or \cite[Lemma 2.2,  Proposition 2.3]{MR3004556}). In the upcoming examples we show that, for very common distributions, the property in \eqref{45} is actually not verified. 
\begin{example}[Poisson process] 
The moments of the Poisson process $N(t)$ with intensity parameter $\lambda$ can be expressed in terms of the so-called \emph{Touchard polynomials}:
$$
\bE[N(1)^n]=T_n(\lambda)=\sum_{k=0}^n \begin{Bmatrix} n\\k\end{Bmatrix} \lambda^k,
$$
where $\begin{Bmatrix} n\\k\end{Bmatrix}$ is the Stirling number of second kind, namely the number of ways to partition a set of $n$ labelled objects into $k$ nonempty unlabelled subsets. Observe that here the cumulants are represented by $\lambda$. If we take, for instance, $\lambda$ to be equal to 1, then  
$$
\bE[N(1)^n]=T_n(1)=\sum_{k=0}^n \begin{Bmatrix} n\\k\end{Bmatrix}=:B_n,
$$
where we denote by $B_n$ the $n$-th Bell number (cf. \cite{riordan}). By computing  the ratio in \eqref{45} numerically, it turns out that the growth of such Bell numbers is actually faster, in the sense that
$$
\limsup_{n\to\infty} \frac 1n\, \frac{\bE\left[N(1)^{2n}\right]}{\bE\left[N(1)^{2n-2}\right]}=\limsup_{n\to\infty} \frac 1n\, \frac{B_{2n}}{B_{2n-2}}=\infty.
$$
\end{example}

\begin{example}[Gamma and Poisson subordinators]
A popular way to construct L\'{e}vy-based models in financial modeling is Brownian subordination (cf. e.g. \cite[Section 4.4]{MR2042661}). This consists in changing the time of a Brownian motion by an independent subordinator, i.e. a increasing L\'{e}vy process. 
Let $W$ be a Brownian motion and $U$ a subordinator, mutually independent. Then, a classical result is that $L(t):=W(U(t))$ is a L\'{e}vy process. The moments of $L$ can be computed in terms of the ones of $W$ and $U$ by the following:
$$
\bE[L(t)^{2n}]=\bE[W(U(t))^{2n}]=\bE[\bE[W(U(t))^{2n}|U]]
=\bE[U(t)^{n}]\bE[W(1)^{2n}].
$$
Thus, 
$$
\frac{\bE\left[L(1)^{2n}\right]}{\bE\left[L(1)^{2n-2}\right]}= \frac{\bE[U(1)^{n}]\,\bE[W(1)^{2n}]}{\bE[U(1)^{n-1}]\,\bE[W(1)^{2n-2}]}\leq
 C_0\,n\, \frac{\bE[U(1)^{n}]}{\bE[U(1)^{n-1}]},
$$
thanks to the properties of Gaussian moments. Consequently, the asymptotic condition in \eqref{45} boils down to the following requirement for the moments of the subordinator:  
\begin{equation}\label{51}
\limsup_{n\to\infty} \frac{\bE[U(1)^{n}]}{\bE[U(1)^{n-1}]}
\leq C.
\end{equation}
However, if $U$ is a Poisson subordinator, then (see previous example) 
$$
\frac{\bE[U(1)^{n}]}{\bE[U(1)^{n-1}]}=\frac{B_n}{B_{n-1}},
$$
which numerically turns out to be unbounded as $n$ tends to infinity. Also, if $U$ is a Gamma subordinator, i.e. $U(1)$ is Gamma distributed, then 
$$
\frac{\bE[U(1)^{n}]}{\bE[U(1)^{n-1}]}\approx C_1\frac{n!}{\beta^{n}}\frac{\beta^{n-1}}{(n-1)!}\approx C_2\, n,
$$
so that, even in this case, property (P) is not satisfied.
\end{example}

\subsection{Risk premium}

We close this section by discussing very briefly the \emph{risk premium}, which  represents a relevant quantity in commodity markets (see e.g. \cite{geman2005}). It is defined as the difference between the forward price and the spot price prediction at delivery, which means, in mathematical terms, 
\begin{equation}\label{75}
\mbox{RP}^f(t,T)=f(t,T)-\bE[f(T,T)|\mathcal{F}_t],
\end{equation}
with the spot price being $S(T):=f(T,T)$. As $f(t,T)$ is a vector, the risk premium is defined as a vector as well. Observe that, under the assumptions of Theorem \ref{85}, we can also write for a martingale measure $\bQ$: 
\begin{equation}\label{73}
\mbox{RP}^f(t,T)=\bE^\bQ[f(T,T)|\mathcal{F}_t]-\bE[f(T,T)|\mathcal{F}_t].
\end{equation}
The economic interpretation of the risk premium is particularly important and has been studied both empirically and theoretically in the context of energy markets by several authors: see, for instance, \cite{Cartea2008829} for the UK gas market, \cite{MR2870526,bck,bessembinder} for power markets and \cite{ronn2008} for various energy exchanges. As our present purpose is just to identify it in our model, we refer to the literature above for a more detailed description of the financial and mathematical features of the risk premium. 

It is natural to extend the definition of the risk premium to swaps as the difference between the swap price and the expected value of the spot price \emph{weighted} over the delivery period:
\begin{equation}\label{74}
\mbox{RP}^F(t,T_1,T_2)=F(t,T_1,T_2)-\bE\left[\left.\int_{T_1}^{T_2} \widehat{w}(T,T_1,T_2) f(T,T)\,dT\right|\mathcal{F}_t\right].
\end{equation}
Let us omit the mathematical justifications in order to keep this discussion at a simple level. As a direct consequence of \eqref{75} and \eqref{74} (cf. also \cite{MR2870526}), we have the following relation:
$$
\mbox{RP}^F(t,T_1,T_2)=\int_{T_1}^{T_2} \widehat{w}(T,T_1,T_2) \mbox{RP}^f(t,T)\,dT.
$$
This means that the risk premium of a swap is the weighted average of the forward risk premium $\mbox{RP}^f(t,T)$ over the delivery period $[T_1,T_2]$.
By writing \eqref{13} in integral form, we have
\begin{align*}
\mbox{RP}^f(t,T)&=\int_t^T \lambda(s) \bE[f(s,T)|\mathcal{F}_t]\, ds-\int_t^T c(s,T)\, ds,
\end{align*}
which becomes, after using \eqref{95},
\begin{align*}
\mbox{RP}^f(t,T)&=(I_n-e^{-\int_t^T \lambda(u) du})f(t,T)
-\int_t^T e^{-\int_s^T \lambda(u) du} c(s,T)\, ds,
\end{align*}
Finally, since $\widehat{w}$ commutes with all the matrices of the same size, the swap risk premium can be written as
\begin{align*}
\mbox{RP}^F(t,T_1,T_2)&=\int_{T_1}^{T_2} (I_n-e^{-\int_t^T \lambda(u) du})\, \widehat{w}(T,T_1,T_2)\,f(t,T)\,dT\\
&\qquad-\int_{T_1}^{T_2} \widehat{w}(T,T_1,T_2) \int_t^T e^{-\int_s^T \lambda(u) du} c(s,T)\, ds\, dT.
\end{align*}

In commodity markets, especially power, the behavior of the observed risk premium is rather involved (see e.g. \cite{weron2008market,bck}). Therefore, in order to suitably describe these stylized facts, a stochastic specification has been postulated. To our knowledge, several models studied in literature, with (at least) the exceptions of \cite{bol,scotti}, imply a deterministic structure of risk premium. Indeed, what we have in the present modeling framework is a stochastically sign-changing (in each component) and affine risk premium.

\section{A Generalization of the Lucia-Schwartz Model}

A classical approach in the literature of energy markets consists in modeling directly the spot price dynamics: see, for instance, \cite{MR2323278,carteaF,gemanR,Lucia2002} and many others. 
This line of research has some advantages: e.g. in order to represent the forward price it is sufficient to take the discounted expected value of the spot under any equivalent probability measure, called the \emph{pricing measure}. In this section we see a  one-dimensional example of spot price dynamics which 
gives forward prices of type \eqref{4}. We start with the original Lucia-Schwartz model and then present a generalization based on a finer representation of the volatility term structure.

The two-factor Lucia-Schwartz \cite{Lucia2002} model consists of two state variables and a seasonal component. If $S(t)$ denotes the spot price at time $t$, then $S(t)=s(t)+X_1(t)+X_2(t)$, with $s(t)$ seasonal deterministic component and 
\begin{align}\label{56}
dX_1(t)&=-\kappa X_1(t)\, dt+ v_1\, dW_1^\bQ(t),
 \\ \label{57}
dX_2(t)&= \mu\, dt + v_2\, dW_2^\bQ(t),
\end{align} 
where $W_1^\bQ$ and $W_2^\bQ$ are two Brownian motions under a pricing measure $\bQ$. Then, we have, for a time $t$ and a maturity date $T\geq t$, 
\begin{equation} \label{spotLS}
S(T)=s(T)+e^{-\kappa (T-t)} X_1(t)+v_1 \int_t^T e^{-\kappa (u-t)} dW_1^\bQ(u)+X_2(t)+\mu(T-t)+v_2(W_2^\bQ(T)-W_2^\bQ(t)).
\end{equation}
Denoting $X(t)=(X_1(t),X_2(t))^\intercal$, the value of a forward contract at time $t$ with delivery at time $T$ is given by
$$
f(t,T)=\bE_\bQ[S(T)|\mathcal{F}_t]=\alpha(t,T)X(t)+\beta(t,T),
$$ 
with
$$
\alpha(t,T)=(e^{-\kappa (T-t)},1) ,\qquad \beta(t,T)=s(T)+\mu (T-t).
$$
By comparing these identities to \eqref{6} and \eqref{8}, we see that this corresponds in our model to
$$
\gamma_0(t)\equiv\gamma_0:=\begin{pmatrix}
0 \\ \mu
\end{pmatrix} ,\qquad 
\gamma_1(t)\equiv\gamma_1:=\begin{pmatrix}
-\kappa & 0 \\
0 & 0\\
\end{pmatrix}.
$$
The relations among the $\bP$-dynamics of $X$ in \eqref{89} and the Girsanov kernel of $\bQ$ are embedded in \eqref{9} and \eqref{10}. In this model, since
$$
v(t)\equiv v:=\begin{pmatrix}
v_1 & 0 \\
0 & v_2\\
\end{pmatrix},
$$
we have that
\begin{align}
v \phi_1(t)&=\gamma_1-\Theta(t), \\
v \phi_0(t)&=\gamma_0-\gamma_1\theta(t),
\end{align}
where $\phi_1(t),\Theta(t)$ are $2\times 2$ matrices and $\phi_0(t),\theta(t)$ are two-dimensional column vectors. By writing them down explicitly, we have
$$
\phi_1(t)=
\begin{pmatrix}
-\frac{\kappa+\Theta_{11}(t)}{v_1} & -\frac{\Theta_{12}(t)}{v_1}\\
& \\
-\frac{\Theta_{21}(t)}{v_2} & -\frac{\Theta_{22}(t)}{v_2}\\
\end{pmatrix},
\quad
\quad
\phi_0(t)=\begin{pmatrix}
\kappa\,\theta_1(t)/v_1 \\ \mu/v_2
\end{pmatrix}.
$$
Then, the coefficients of the $\bP$-dynamics of $f(t,T)$ are determined by \eqref{23}--\eqref{25}, which yield
\begin{align*}
\lambda(t)\alpha(t,T)&=\alpha(t,T)(\gamma_1(t)-\Theta(t))
=\alpha(t,T)v(t)\phi_1(t),
\\
c(t,T)&=\lambda(t)\beta(t,T)+\alpha(t,T)(\theta(t)-\gamma_0(t)),
\\
\sigma(t,T)&=\alpha(t,T)v(t).
\end{align*}
We immediately observe that $\alpha(t,T)$ admits a right inverse for each $T$ and there are infinitely many $\lambda(t)$ satisfying the first relation. For instance, we could take (by assuming that $\Theta_{12}=\Theta_{21}=0$)
\begin{align*}
\lambda(t)&=-\kappa-\Theta_{11}(t),
\\
c(t,T)&=\lambda(t)(s(T)+\mu(T-t))+
\theta_1(t) e^{-\kappa(T-t)}+\theta_2(t)-\mu,
\\
\sigma(t,T)&=(e^{-\kappa(T-t)}v_1,v_2).
\end{align*}

In view of this analysis, we now present a two-factor model which can be seen as a generalization of \cite{Lucia2002}, where we add a term structure for the volatility. An empirical study of this model has been performed in a companion paper by \cite{latini}. 

By reformulating \eqref{56} and \eqref{57}, we define 
\begin{align}\label{56N}
dX_1(t)&=-\kappa X_1(t)\, dt+ v_1\, dW_1^\bQ(t),
 \\ \label{57N}
dX_2(t)&= \frac{v'_2(t)}{v_2(t)} X_2(t)\, dt + v_2(t)\, dW_2^\bQ(t),
\end{align} 
where $v_2(t)$ is now a differentiable function of time such that $v_2(t) > 0$ for all $t$ and that $\frac{v_2'(t)}{v_2(t)} \neq - \kappa$ (otherwise, the model collapses into a 1-factor model). In other words, $X_1$ is a mean-reverting Ornstein-Uhlenbeck process, while $X_2(t)/v_2(t)$ can be shown to be a Brownian motion. While in the Lucia-Schwartz model $X_2$ is a Brownian motion with drift as $v_2$ is constant, here the variance of $X_2$ is varying in time, which allows to have seasonality in the volatility of $S$ as well as in the price level. Thus we can regard $v_2$ as the seasonality factor for the volatility and $s$ as the seasonality factor for the price. By computations analogous to those before, we derive that
$$
S(T)=s(T)+e^{-\kappa (T-t)} X_1(t)+v_1 \int_t^T e^{-\kappa (u-t)} dW_1^\bQ(u)+\frac{v_2(T)}{v_2(t)} X_2(t)+v_2(T)(W_2^\bQ(T)-W_2^\bQ(t)).
$$
Notice that, within this new formulation, the processes $X_1$ and $X_2$ are stationary in mean (i.e., $\bE[X_1(t)] = X_1(0)$ and $\bE[X_2(t)] = X_2(0)$ for all $t$). Thus, if we start with $X_1(0) = X_2(0) = 0$, then $s(T)$ is also the expectation of $S(T)$ (under the risk-neutral measure $\bQ$). 

The value at time $t$ of a forward contract with delivery at time $T$ is 
$$
f(t,T)=\bE_\bQ[S(T)|\mathcal{F}_t]=\alpha(t,T)X(t)+\beta(t,T),
$$ 
where now
$$
\alpha(t,T)= \left(e^{-\kappa (T-t)},\frac{v_2(T)}{v_2(t)}\right) ,\qquad \beta(t,T)=s(T).
$$
This implies
\begin{equation} \label{fLS} 
df(t,T) = e^{- \kappa (T-t)} \sigma_1\ dW^\bQ_1(t) + \sigma_2(T)\ dW^\bQ_2(t),
\end{equation}
so that, in particular, $ \sigma(t,T) = (e^{- \kappa (T-t)} \sigma_1, \sigma_2(T))$. This model allows for the instantaneous forward contracts $f(t,T)$ to have a term structure of the volatility which accounts both for the Samuelson effect (volatility increasing as $t \to T$) in the term $e^{- \kappa (T-t)} \sigma_1$, as well as a potentially complex seasonality with respect to the absolute maturity in the term $\sigma_2(T)$. 
From this, by specifying a choice of the weight function $\hat w$ and using \eqref{21}, we can obtain the dynamics of $F(t,T_1,T_2)$ for all $(t,T_1,T_2) \in \mathcal{A}^{\mathbb{T}}_2$. For example, if $\hat w \equiv \frac{1}{T_2-T_1}$, we have that 
$$ dF(t,T_1,T_2) = - \frac{e^{- \kappa (T_2-t)} - e^{- \kappa (T_1-t)}}{\kappa(T_2 - T_1)} \sigma_1\ dW^\bQ_1(t) + \frac{\int_{T_1}^{T_2} \sigma_2(T)\ dT}{T_2 - T_1} \ dW^\bQ_2(t). $$

If we want that, under the empirical measure, $f$ follows a mean-reverting process as in \eqref{13}, we have to relate   
the coefficients of $X$ in its $\bP$-dynamics \eqref{89} to \eqref{23}--\eqref{25}. We already know that, both under $\bP$ as under $\bQ$, we have
$$
v(t)=\begin{pmatrix}
v_1 & 0 \\
0 & v_2(t)\\
\end{pmatrix}. 
$$
Similarly to what we did before for the Lucia-Schwartz model, we can make a parsimonious choice and put $\theta(t) \equiv 0$ and $\Theta_{12}(t) = \Theta_{21}(t) \equiv 0$. This gives us 
$$ c(t,T) = \lambda(t) s(T) $$
and
$$
\Theta(t)=\begin{pmatrix}
\Theta_{11}(t) & 0 \\
0 & \Theta_{22}(t)\\
\end{pmatrix}  = \begin{pmatrix}
- \kappa - \lambda(t) & 0 \\
0 & \frac{v'_2(t)}{v_2(t)} - \lambda(t)\\
\end{pmatrix}.
$$
Regarding the market price of risk 
we observe that
$$
\gamma_0(t)\equiv0=\begin{pmatrix}
0 \\ 0
\end{pmatrix} ,\qquad 
\gamma_1(t)=\begin{pmatrix}
-\kappa & 0 \\
0 & \frac{v'_2(t)}{v_2(t)}\\
\end{pmatrix}
$$
and 
%
$$ v(t) \phi_1(t) = \gamma_1(t) - \Theta(t) = \lambda(t) I_2,$$
where $I_2$ is the 2 $\times$ 2 identity matrix. This, together with \eqref{10}, gives 
$$ \phi_1(t) = \left( \frac{\lambda(t)}{v_1}, \frac{\lambda(t)}{v_2(t)} \right), \qquad \phi_0(t) \equiv 0. $$
We now follow Theorem \ref{7} to determine if the process $\phi$ obtained above gives a martingale. By the sufficient conditions there, if we impose that $\Theta$ is bounded, i.e. that $\lambda$ and $\frac{v'_2}{v_2}$ are bounded, and that $\phi_1$ is bounded, i.e. that $\frac{\lambda}{v_2}$ is bounded, then $Z$ is a martingale, and $\bQ$ is an equivalent martingale measure. Thus, the $\bQ$-dynamics in \eqref{fLS} corresponds to a $\bP$-dynamics
\begin{equation} \label{FLS} 
df(t,T) = \lambda(t)( s(T) - f(t,T))\ dt + e^{- \kappa (T-t)} \sigma_1\ dW_1(t) + \sigma_2(T)\ dW_2(t) 
\end{equation}
and, if for example $\hat w \equiv \frac{1}{T_2-T_1}$, we also have that 
\begin{align}
dF(t,T_1,T_2) = \quad & \lambda(t) \left( \frac{\int_{T_1}^{T_2} s(T)\ dT}{T_2 - T_1} - F(t,T_1,T_2)\right)\ dt \\
& - \frac{e^{- \kappa (T_2-t)} - e^{- \kappa (T_1-t)}}{\kappa(T_2 - T_1)} \sigma_1\ dW_1(t) + \frac{\int_{T_1}^{T_2} \sigma_2(T)\ dT}{T_2 - T_1} \ dW_2(t).
\end{align}

From this model, we can recover various stylized facts typical of energy markets.
\begin{itemize}
\item The dynamics of forward prices $f(\cdot,T)$ under the real-world probability measure $\bP$ is mean-reverting. The mean-reversion speed $\lambda(t)$ can be time-dependent (but not maturity-dependent in this formulation), and the long-term mean $s(T)$ is exactly the seasonal component of the spot price $S$. The same can be said about the swap dynamics $F(\cdot,T_1,T_2)$, where this time the long-term mean is the maturity-average of the long-term mean $s(T)$, $T \in [T_1,T_2]$. 
\item We have a generalized Samuelson effect in the forward prices $f(\cdot,T)$, which is quite evident, and also in the swap prices $F(\cdot,T_1,T_2)$. About this latter, denote the diffusion vector of $F(t,T_1,T_2)$ as $\Sigma(t,T_1,T_2) = (e^{\kappa t} \Gamma(T_1,T_2), \Psi(T_1,T_2))$, with
\begin{eqnarray*}
\Gamma(T_1,T_2) & := & \frac{1}{T_2-T_1} \int_{T_1}^{T_2} \sigma_1 e^{-\kappa u} \, du = 
\frac{\sigma_1(e^{-\kappa T_1}-e^{-\kappa T_2})}{\kappa (T_2-T_1)}, \\
\Psi(T_1,T_2) & := & \frac{1}{T_2-T_1} \int_{T_1}^{T_2} \sigma_2(u) \, du.
\end{eqnarray*}
then the function $t \to \|\Sigma(t,T_1,T_2)\|$ is increasing in $t$: as time to maturity decreases, the volatility increases. 
\item Swap prices with shorter delivery periods are more volatile than swap prices with longer delivery periods. In fact, $|\Gamma(T_1,T_2)|$ is decreasing in the second variable. This means that, for all $T_1 < T_2 < T_3$, we have $|\Gamma(T_1,T_2)| > |\Gamma(T_1,T_3)|$. For instance, being equal the time to maturity, monthly contracts are more volatile than quarters (lasting 3 months) or calendars (lasting 1 year). This is in line with empirical findings, see e.g. \cite{Benth20081116}, where the authors perform an empirical analysis of electricity contracts traded on Nord Pool.
\end{itemize} 

In conclusion, here we have shown that a natural extension of Lucia and Schwartz could be to allow for a specific time-dependent speed of mean-reversion in the second factor, such that $X_2/v_2(t)$ is a Brownian motion, as well as for a time-dependent speed of mean reversion of forward prices. Also, we worked out the ingredients for this model in our framework, so that the derived dynamics is arbitrage free.

\section{A Two-Commodity Cointegrated Market}

Cointegration is a well-known concept in econometrics and indicates a phenomenon observed in several energy-related markets (see, for instance, \cite{alexander1999correlation,de2009cointegration}). Let us introduce an arbitrage-free cross-commodity model with mean-reversion, which accounts for cointegrated price movements. We present it in the case of two commodities for the sake of simplicity, as generalizations are straightforward. We are inspired by the spot price model presented in \cite{MR3630550}. After specifying the spot price dynamics under a martingale measure $\bQ$, we will derive forward prices by the conditional expectation. 
The focus in this example is to explore the possibility to have mean-reverting cointegration, which will be shown to lead to some interesting model restrictions.

Let $S_k(t)$ denote the spot price at time $t$ of commodity $k$ for $k=1,2$ and set 
\begin{equation}\label{79}
S_k(t)=s_k(t)+Y_k(t)+a_k L(t),
\end{equation}
with $s_k(t)$ seasonal deterministic component, $a_k\in\bR\setminus\{0\}$ for $k=1,2$. We assume that, under a martingale measure $\bQ$, the dynamics of these factors are 
\begin{align*}
dL(t)&=\sigma dW(t) + \psi dJ(t),
\\ 
dY_1(t)&=-\mu_1 Y_1(t)\, dt+ \sigma_1\, dW_1(t)+ \psi_1 dJ_1(t),
 \\ 
dY_2(t)&= -\mu_2 Y_2(t)\, dt + \sigma_2\, dW_2(t)+ \psi_2 dJ_2(t),
\end{align*} 
where $\sigma,\psi,\mu_k,\sigma_k,\psi_k\in\bR$ for $k=1,2$. The process $(W,W_1,W_2)$ is a three-dimensional Brownian motion independent from $(J,J_1,J_2)$, which is a three-dimensional zero-mean pure-jump Lévy process with independent components. The Lévy measures associated are denoted by $\nu,\nu_1$ and $\nu_2$ and satisfy the integrability assumptions in Theorem \ref{85}. Since $S_2(t)/a_2-S_1(t)/a_1=s_2(t)/a_2-s_1(t)/a_1+Y_2(t)/a_2-Y_1(t)/a_1$, we say that $S_1$ and $S_2$ are cointegrated around the seasonality function $s_2(t)/a_2-s_1(t)/a_1$ (cf. \cite{MR3630550}).

Consistently with the notation in Section 3, we have that $\kappa=\int_\bR y^2\,\nu(dy)$ and $\kappa_i=\int_\bR y^2\,\nu_i(dy)$, $i=1,2$, and 
\begin{align*}
v(t)\equiv v:=\begin{pmatrix}
\sigma &0 & 0\\
0 & \sigma_1 & 0 \\
0 & 0 & \sigma_2 \\
\end{pmatrix},\quad
z(t)\equiv z:=\begin{pmatrix}
\psi &0 & 0\\
0 & \psi_1 & 0 \\
0 & 0 & \psi_2 \\
\end{pmatrix},\quad
K=\begin{pmatrix}
\kappa &0 & 0\\
0 & \kappa_1 & 0 \\
0 & 0 & \kappa_2 \\
\end{pmatrix}.
\end{align*}
Then, the value of a forward contract for the commodity $k$ at time $t$ with delivery at time $T$ is given by the conditional expectation
$$
f_k(t,T)=\bE_\bQ[S_k(T)|\mathcal{F}_t]=s_k(T)+e^{-\mu_k (T-t)} Y_k(t)+a_k L(t).
$$
We denote $S(t)=(S_1(t),S_2(t))^\intercal$, $f(t,T)=(f_1(t,T),f_2(t,T))^\intercal$, $X(t)=(L(t),Y_1(t),Y_2(t))^\intercal$, so that
$$
f(t,T)=\bE_\bQ[S(T)|\mathcal{F}_t]=\alpha(t,T)X(t)+\beta(t,T),
$$ 
with
\begin{align}\label{106}
\alpha(t,T)=\begin{pmatrix}
a_1 & e^{-\mu_1 (T-t)} & 0\\
a_2 & 0 & e^{-\mu_2 (T-t)}
\end{pmatrix},
\qquad
\beta(t,T)=\begin{pmatrix}
s_1(T)\\ 
s_2(T)
\end{pmatrix}.
\end{align}
Besides, by comparing \eqref{106} to \eqref{6} and \eqref{8}, we have
$$
\gamma_0(t)\equiv\begin{pmatrix}
0 \\ 0
\end{pmatrix} ,\qquad 
\gamma_1(t)\equiv\gamma_1:=\begin{pmatrix}
0 &0 & 0\\
0 & -\mu_1 & 0 \\
0 & 0 & -\mu_2 \\
\end{pmatrix}.
$$
Following Proposition \ref{60} and Theorem \ref{85}, the martingale measures  $\bQ$ are determined by the kernels $\phi(t)=\phi_1(t)X(t)+\phi_0(t)$ and $\xi(t)=\xi_0(t)$, where $\phi_1,\phi_0$ and $\xi$ satisfy
\begin{align*}
\gamma_1&=v\phi_1(t)+\Theta(t),\\ 
0&=v \phi_0(t)+z K\xi(t)+\theta(t).
\end{align*}
The deterministic functions $\phi_1(t),\Theta(t)$ are $3\times 3$ matrices and $\phi_0(t),\xi(t),\theta(t)$ are three-dimensional column vectors. The coefficients of the $\bP$-dynamics of $f(t,T)$ are determined by
\begin{align*}
\lambda(t)\alpha(t,T)=\alpha(t,T)(\gamma_1-\Theta(t))
&=\alpha(t,T)v \phi_1(t),
\\
c(t,T)=\lambda(t)\beta(t,T)+\alpha(t,T)\theta(t)&=\lambda(t)\beta(t,T)-\alpha(t,T)(v\phi_0(t)+z K\xi(t)),
\\
\sigma(t,T)&=\alpha(t,T)v
\\
\psi(t,T)&=\alpha(t,T)z.
\end{align*}
In particular, in view of Proposition \ref{27}, we must verify under which conditions it is possible to define a $2\times 2$ matrix 
$$
\lambda(t):=\begin{pmatrix}
\lambda_{11}(t) & \lambda_{12}(t) \\
\lambda_{21}(t) & \lambda_{22}(t)\\
\end{pmatrix},
$$ 
satisfying (CP) (Definition \ref{17}), such that, independently from $T$, 
\begin{equation}\label{107}
\lambda(t)\alpha(t,T)=\alpha(t,T)M(t),
\end{equation}
where we denote $M(t)=v \phi_1(t)$.
From \eqref{107} we get the following system of equations:
\begin{align*}
a_1 \lambda_{11}(t)+a_2 \lambda_{12}(t) &= a_1 M_{11}(t) + e^{-\mu_1(T-t)} M_{21}(t),\\
e^{-\mu_1(T-t)} \lambda_{11}(t) &= a_1 M_{12}(t) + e^{-\mu_1(T-t)} M_{22}(t),\\
e^{-\mu_2(T-t)} \lambda_{12}(t) &= a_1 M_{13}(t) + e^{-\mu_1(T-t)} M_{23}(t),\\
a_1 \lambda_{21}(t)+a_2 \lambda_{22}(t) &= a_2 M_{11}(t) + e^{-\mu_2(T-t)} M_{31}(t),\\
e^{-\mu_1(T-t)} \lambda_{21}(t) &= a_2 M_{12}(t) + e^{-\mu_2(T-t)} M_{32}(t),\\
e^{-\mu_2(T-t)} \lambda_{22}(t) &= a_2 M_{13}(t) + e^{-\mu_2(T-t)} M_{33}(t),
\end{align*}
which admits a unique solution if and only if
\begin{equation}\label{107bis}
\mu_1=\mu_2
\end{equation}
and
\begin{equation}\label{80}
\lambda_{11}(t)+\frac{a_2}{a_1} \lambda_{12}(t) = \frac{a_1}{a_2} \lambda_{21}(t)+\lambda_{22}(t).
\end{equation}
Furthermore, for the market price of risk parameters we have 
$$
\phi_1(t)=\begin{pmatrix}
\frac{a_1\lambda_{11}(t)+a_2 \lambda_{12}(t)}{a_1 v}  & 0 & 0 \\
0 & \frac{\lambda_{11}(t)}{v_1}& \frac{\lambda_{12}(t)}{v_1}\\
0 & \frac{\lambda_{21}(t)}{v_2}& \frac{\lambda_{22}(t)}{v_2}\\
\end{pmatrix},
$$
while $\phi_0(t)$ and $\xi(t)$ must satisfy
$$
v \phi_0(t)+z K\xi(t)=-\theta(t).
$$

Consequently, if we start with a spot dynamics as in \eqref{79}, with $\mu:=\mu_1=\mu_2$ and a matrix $\lambda$ satisfying the condition in \eqref{80} and (CP), then we can build the following mean-reverting, arbitrage-free, cointegrated $\bP$-dynamics for the two forwards:
\begin{eqnarray*}
df_1(t,T)&=&(c_1(t,T)-\lambda_{11}(t) f_1(t,T)-\lambda_{12}(t) f_2(t,T))\, dt\\
&&+a_1 \sigma dW(t)+a_1 \psi dJ(t)+e^{-\mu (T-t)} \sigma_1 dW_1(t)+e^{-\mu (T-t)} \psi_1 dJ_1(t),\\
df_2(t,T)&=&(c_2(t,T)-\lambda_{21}(t) f_1(t,T)-\lambda_{22}(t) f_2(t,T))\, dt\\
&&+a_2 \sigma dW(t)+a_2 \psi dJ(t)+e^{-\mu (T-t)} \sigma_2 dW_2(t)+e^{-\mu (T-t)} \psi_2 dJ_2(t).
\end{eqnarray*}
These features are naturally inherited by the swap contracts $F(t,T_1,T_2)$:
\begin{eqnarray*}
dF_1(t,T_1,T_2)&=&(C_1(t,T_1,T_2)-\lambda_{11}(t) F_1(t,T_1,T_2)-\lambda_{12}(t) F_2(t,T_1,T_2))\, dt\\
&&+a_1 \sigma dW(t)+ \frac{\sigma_1(e^{-\mu T_1}-e^{-\mu T_2})}{\mu (T_2-T_1)} dW_1(t)+ \frac{\psi_1(e^{-\mu T_1}-e^{-\mu T_2})}{\mu (T_2-T_1)} dJ_1(t),\\
dF_2(t,T_1,T_2)&=&(C_2(t,T_1,T_2)-\lambda_{21}(t) F_1(t,T_1,T_2)-\lambda_{22}(t) F_2(t,T_1,T_2))\, dt\\
&&+a_2 \sigma dW(t)+\frac{\sigma_2(e^{-\mu T_1}-e^{-\mu T_2})}{\mu (T_2-T_1)} dW_2(t)+ \frac{\psi_2(e^{-\mu T_1}-e^{-\mu T_2})}{\mu (T_2-T_1)} dJ_2(t),
\end{eqnarray*}
where we made the choice $\hat w \equiv \frac{1}{T_2-T_1} I_2$. 

\begin{remark}
The structural Equation \eqref{80} for the mean-reversion coefficients $\lambda_{ij}$ allows to have much flexibility for modeling. For example, in such a setting, one could model the price of the first commodity with a Markovian dynamics  (i.e. $\lambda_{12}(t)\equiv0$) and preserve in the drift of commodity 2 the dependence on commodity 1. To cite an example of application, this can reproduce the behavior of oil (commodity 1) and gas (commodity 2) prices. 
\end{remark}

This model shows how to incorporate cointegration into an arbitrage-free, mean-reverting forward market with two commodities. It is interesting to observe that the no-arbitrage constraints imply a structural condition on the shape of the volatility term structure \eqref{107bis} and a linear relation among the mean-reversion coefficients  \eqref{80}. We furthermore remark that a similar approach allows to design more flexible market models, as well as accounting for more than two commodities, still preserving the mean-reverting and arbitrage-free traits. 

\section{Conclusions}

By adapting the Heath-Jarrow-Morton idea to energy forward curves in additive models, we have introduced an arbitrage-free framework capable of producing flexible and tractable market models, which exhibit mean-reversion in the dynamics of forward prices under the real-world probability measure.

Our main assumption on the functional form of the forward price processes has allowed to solve several issues. The fundamental requirement for no-arbitrage is the existence of an equivalent martingale measure. Generally, finding it is not a trivial task, since valid measure changes must be independent of the delivery parameters. Mean-reversion naturally requires affine Girsanov kernels, which are in particular stochastic and unbounded. We have been able to validate under minimal assumptions rather general density processes, with stochastic kernels for the Brownian components and Esscher-type factors for the pure-jump part. 
In addition, swaps and forwards have to satisfy an integral relation, which in general leads to losing analytical tractability, and in particular the Markov property. We succeed in preserving them by introducing simple relations among the coefficients of the dynamics. 

Passing to the applications, first we have shown that the well-known Lucia-Schwarz model was already included in the class of models that we here characterized. We have also presented an extension of it, capable to model seasonality in forward's volatilities, besides the price seasonality component already present in the original model. 

We then presented a multidimensional market model, which has enabled us in particular to reproduce cointegration effects. As the energy-related markets are strongly interconnected for physical reasons, the possibility of modeling dependence relations more sophisticated than correlation is particularly important for our application purposes. 

Looking ahead to future research, since the additive dynamics produce tractable price processes, we believe to have opened the way for analytical formulas for complex derivatives, multicommodity portfolios and risk measures. 

\appendix\normalsize

\section{Proof of Theorem \ref{85}}

Part of this proof will be presented in a simpler version than the one developed in an earlier draft of this work, thanks to the useful suggestions of an anonymous referee. Introduce the sequence of stopping times $\tau_n:=\inf\{t\geq0: \|X(t) \|\geq n\}$. Set $f(y):=(1+y)\log(1+y)-y$ and define the process:
$$
B(t):=\frac12 \int_0^t \|\phi(s)\|^2 \, ds + \int_0^t f(\xi(s)^\intercal y) \,\nu(dy),
$$
which is the predictable compensator of 
$$
\frac12 \langle H^c,H^c \rangle + \sum_{t\leq\cdot} f(\Delta H(t)).
$$
In particular, observe that, for every $n\in\bN$, the stopped process $B^{\tau_n}(t):=B(t\wedge \tau_n)$ is bounded. From \cite[Theorem III.1]{MR489492} it follows that $Z^{\tau_n}$ is a uniformly integrable martingale.  Define the probability measure $\bQ^n$ by setting 
$$
\frac{d\bQ^n}{d\bP}:=Z^{\tau_n}(\bT).
$$
Then, we have
\begin{align*}
\bE[Z(\bT)]&\geq \bE[Z(\bT)\mathbbm{1}_{\tau_n >\bT}]=\bE[Z^{\tau_n}(\bT)\mathbbm{1}_{\tau_n >\bT}]=\bQ^n(\tau_n>\bT)\\
&=1-\bQ^n\left(\sup_{t\in[0,\bT]} \|X(t) \|\geq n\right)\geq 1-\frac{\bE_{\bQ^n}\left[\sup_{t\in[0,\bT]} \|X(t) \|\right]}{n}.
\end{align*}
If we show that 
\begin{equation}\label{sup}
\bE_{\bQ^n}\left[\sup_{t\in[0,\bT]} \|X(t) \|\right]\leq C,
\end{equation}
for a constant $C$ independent on $n$, then $\bE[Z(\bT)]\geq 1-C/n$ for all $n\in\bN$, which implies that $\bE[Z(\bT)]=1$. As a consequence, in order to conclude the proof it is sufficient to verify \eqref{sup}. 

By Girsanov's Theorem, it holds that 
\begin{align*}
dX(s)=\bigl(&\theta(s)+v(s)\phi_0(s) \mathbbm{1}_{[0,\tau_n]}(s)+z(s) K\,\xi(s)\mathbbm{1}_{[0,\tau_n]}(s)+\Theta(s)X(s)\\
&+v(s)\phi_1(s) X(s) \mathbbm{1}_{[0,\tau_n]}(s)\bigr)\,ds+ v(s)dW^n(s)+ z(s)\int_{\bR^k}y\overline N^n(ds,dy),
\end{align*}
where, under the measure $\bQ^n$, $W^n(s)=W(s)-\int_0^s \mathbbm{1}_{[0,\tau_n]}(u)\phi(u)\, du$ is a Brownian motion and $\overline N^n(ds,dy)=\overline N(ds,dy)-\mathbbm{1}_{[0,\tau_n]}(s)\,\xi(s)^\intercal y\,\nu(dy)ds$ is the $\bQ^n$-compensated Poisson random measure of $N$ (see \cite[Theorem 1.35]{MR2322248}). 
Hence, for each $t\in[0,\bT]$, (set w.l.o.g. $X(0)=0$)
\begin{align*}
\bE^{\bQ^n}&\left[\sup_{s\in[0,t]}\left\|X(s)\right\|^2\right]\\
&\leq 4\,\Biggl( \bE^{\bQ^n}\left[\sup_{s\in[0,t]} \left\| \int_0^s \bigl(\theta(u)+v(u)\phi_0(u) \mathbbm{1}_{[0,\tau_n]}(u)+z(u) K\,\xi(u)\mathbbm{1}_{[0,\tau_n]}(u)\bigr)\,du \right\|^2\right]\\
&+ \bE^{\bQ^n}\left[\sup_{s\in[0,t]} \left\|\int_0^s \bigl(\Theta(u)+\mathbbm{1}_{[0,\tau_n]}(u) v(u)\phi_1(u)\bigr) X(u)\,du \right\|^2\right]\\
&+\bE^{\bQ^n}\left[\sup_{s\in[0,t]} \left\|\int_0^s v(u)dW^n(u)\right\|^2\,\right] 
+ \bE^{\bQ^n}\left[\sup_{s\in[0,t]} \left\|\int_0^s \int_{\bR^k} z(u)\,y\,\overline N^n(du,dy)\right\|^2\,\right]\Biggr).
\end{align*}
Firstly, by integrability assumptions on the coefficients, we have
\begin{align*}
&\bE^{\bQ^n}\left[\sup_{s\in[0,t]} \left\| \int_0^s \bigl(\theta(u)+v(u)\phi_0(u) \mathbbm{1}_{[0,\tau_n]}(u)+z(u) K\,\xi(u)\mathbbm{1}_{[0,\tau_n]}(u)\bigr)\,du \right\|^2\right]\\
&+\,\bE^{\bQ^n}\left[\sup_{s\in[0,t]} \left\|\int_0^s \bigl(\Theta(u)+\mathbbm{1}_{[0,\tau_n]}(u) v(u)\phi_1(u)\bigr) X(u)\,du \right\|^2\right]\\
&\leq \bT\,\bE^{\bQ^n}\left[ \int_0^\bT \left\|\theta(u)+v(u)\phi_0(u) \mathbbm{1}_{[0,\tau_n]}(u)+z(u) K\,\xi(u)\mathbbm{1}_{[0,\tau_n]}(u)\right\|^2\,du \right]\\
&+\bT\,\bE^{\bQ^n}\left[\int_0^t \left\|\Theta(u)+\mathbbm{1}_{[0,\tau_n]}(u) v(u)\phi_1(u)\right\|^2 \left\| X(u)\right\|^2\,du \right]\\
&\leq C_1 + C_2 \int_0^t \bE^{\bQ^n}\left[\sup_{v\in[0,u]}\left\|X(v)\right\|^2\right]\, du,
\end{align*}
for some constants $C_1,C_2$ independent of $n$. For the martingale part, 
we apply Doob's $\bL^2$-inequalities (e.g. \cite[Theorem 5.1.3]{MR3443368}). As $v$ is square-integrable, the first term yields
\begin{align*}
\bE^{\bQ^n}\left[\sup_{s\in[0,t]} \left\|\int_0^s v(u)dW^n(u)\right\|^2\,\right] 
\leq\bT\,\bE^{\bQ^n}\left[\int_0^\bT \|v(s)\|^2 ds\,\right]= \bT\,C_3.
\end{align*}
For the second term, since the Lévy measure $\nu$ admits fourth moment, $\xi$ is bounded and $z$ is square-integrable, we have that 
\begin{align*}
\bE^{\bQ^n}&\left[\sup_{s\in[0,t]} \left\|\int_0^s \int_{\bR^k} z(u)\,y\,\overline N^n(du,dy)\right\|^2\,\right]\\
&\leq \bT\,\bE^{\bQ^n}\left[\int_0^\bT \int_{\bR^k} \left\|z(s)\,y+\mathbbm{1}_{[0,\tau_n]}(s)\,z(s)\, y\,\xi(s)^\intercal y \right\|^2\,\nu(dy)ds\,\right]\\
&\leq \bT\,C_4\left(\int_{\bR^k} \|y\|^2 \nu(dy) +\int_{\bR^k} \|y\|^4 \nu(dy)\right),
\end{align*}
with $C_3,C_4$ independent of $n$.
Finally, all the bounds together give
$$
\bE^{\bQ^n}\left[\sup_{s\in[0,t]}\left\|X(s)\right\|^2\right]\leq C_5 + C_6 \int_0^t  \bE^{\bQ^n}\left[\sup_{v\in[0,u]}\left\|X(v)\right\|^2\right]\, du,
$$
again with constants independent of $n$. 
We conclude by Gronwall's lemma that
$$
\bE^{\bQ^n}\left[\sup_{t\in[0,\bT]}\left\|X(t)\right\|^2\right]\leq C_5 \,e^{C_6 \bT},
$$
which implies \eqref{sup}. This concludes the proof.

\section{Proof of Theorem \ref{7}}

The proof is based on the same idea of Proposition 5.1 in \cite{weather}, where a weak Novikov-type condition is applied to the series representation of the exponential. In view of \cite[Corollary 5.14]{KaS:91} it is sufficient to prove that there exists an increasing sequence of positive real numbers $(t_k)_{k\in\bN}$ diverging to $+\infty$ such that, for all $k\in\bN$,
\begin{equation}\label{36}
\bE\left[ \exp\left(\frac12 \int_{t_k}^{t_{k+1}} \|\phi(t)\|^2 dt \right)\right]<\infty,
\end{equation}
where
$$
\phi(t)=\phi_1(t)X(t)+\phi_0(t)
$$
and
\begin{equation}\label{63}
dX(t)=(\theta(t)+\Theta(t)X(t))dt+v(t)dW(t).
\end{equation}
Since $\Theta$ satisfies (CP) (see Definition \ref{17}), we can write the unique solution of \eqref{63} as
\begin{equation*}
\begin{split}
X(t)=e^{\int_0^t \Theta(s) ds} X(0)+\int_0^t e^{\int_s^t \Theta(u) du} \theta(s) ds
+\int_0^t e^{\int_s^t \Theta(u) du} v(s) dW(s).\\
\end{split}
\end{equation*}
Now, we start to estimate the quantity in \eqref{36} by
$$
\frac12 \|\phi(t)\|^2\leq \|\phi_1(t) X(t)\|^2+ \|\phi_0(t)\|^2.
$$
The second term is deterministic, so as to prove \eqref{36} we can neglect it. Hence, we move to consider the first one. Let us denote $U(t):=e^{\int_0^t \Theta(s) ds}$. Then, 
\begin{align*}
\|\phi_1(t) X(t)\|^2 &\leq 3\|\phi_1(t) U(t) X(0)\|^2 + 3\int_0^t \|\phi_1(t) U(t)U(s)^{-1} \theta(s)\|^2 ds\\
&+ 3 \left\|\int_0^t \phi_1(t) U(t)U(s)^{-1} v(s) dW(s)\right\|^2
\\
& \leq 3\|\phi_1(t) U(t)\|^2 \left( \|X(0)\|^2+\int_0^t \|U(s)^{-1} \theta(s)\|^2 ds
+\left\|\int_0^t U(s)^{-1} v(s) dW(s)\right\|^2 \right).
\end{align*}
As $\phi_1$ and $\Theta$ are bounded, we are essentially left with
\begin{equation*}
\bE\left[ \exp\left(\frac12 \int_{t_k}^{t_{k+1}} \|\phi(t)\|^2 dt \right)\right]
\leq
C_1 \bE\left[ \exp\left(C_2 
\int_{t_k}^{t_{k+1}} \left\|\int_0^t U(s)^{-1} v(s) dW(s)\right\|^2 dt \right)\right],
\end{equation*}
where $C_1,C_2$ are constants. 

Let us introduce the matrix-valued function $g(s):=U(s)^{-1} v(s)$. Now, by Lebesgue's theorem
\begin{equation}\label{71}
\bE\left[ \exp\left( C_2 \int_{t_k}^{t_{k+1}}  \left\|\int_0^t g(s) dW(s)\right\|^2 dt \right)\right]
=\sum_{n=0}^\infty \frac{C_2^n}{n!}\bE\left[\left(\int_{t_k}^{t_{k+1}}  \left\|\int_0^t g(s) dW(s)\right\|^2 dt \right)^n\right].
\end{equation}
Next, from Jensen's inequality,
$$
\bE\left[\left(\int_{t_k}^{t_{k+1}}  \left\|\int_0^t g(s) dW(s)\right\|^2 dt \right)^n\right]
\leq (t_{k+1}-t_k)^{n-1}\int_{t_k}^{t_{k+1}}  \bE\left[\left\|\int_0^t g(s) dW(s)\right\|^{2n} \right] dt.
$$
Notice that, since $\int_0^t \|g(s)\|^2 ds<\infty$, the random variable $\int_0^t g(s) dW(s)$  is equal in law to $\left(\sqrt{\int_0^t \|g(s)\|^2 ds}\right) Z$, where $Z$ is a standard normal vector. Therefore, since $g$ is square-integrable,
$$
\bE\left[\left\|\int_0^t g(s) dW(s)\right\|^{2n} \right]=
\left(\int_0^t \|g(s)\|^2 ds\right)^n \bE(\|Z\|^{2n})
\leq C_3^n\,\bE(\|Z\|^{2n}),
$$
for a constant $C_3$. 

Finally, exploiting the last two estimates, we find that \eqref{71} is bounded by
\begin{align*}
&\bE\left[ \exp\left( C_2 \int_{t_k}^{t_{k+1}}  \left\|\int_0^t g(s) dW(s)\right\|^2 dt \right)\right]
\leq \sum_{n=0}^\infty \frac{(C_2 C_3)^n (t_{k+1}-t_k)^{n}}{n!}\bE(\|Z\|^{2n}).
\end{align*}
A well-known property of Gaussian moments is the following
$$
\bE(\|Z\|^{2n})\leq C_4 n \bE(\|Z\|^{2(n-1)}).
$$
By applying the ratio test for series we get that, if we choose $(t_k)_k$ such that $C_2 C_3 C_4 (t_{k+1}-t_k)<1$, then
$$
\sum_{n=0}^\infty \frac{(C_2 C_3 C_4)^n (t_{k+1}-t_k)^n}{n!}\bE(\|Z\|^{2n})<\infty,
$$
which implies \eqref{36} and, therefore, the statement.

\bibliographystyle{amsplain}
\bibliography{bibbase}

\end{document}